\documentclass[10pt,final,reqno]{amsart}
     \makeatletter
     \def\section{\@startsection{section}{1}%
     \z@{.7\linespacing\@plus\linespacing}{.5\linespacing}%
     {\bfseries
     \centering
     }}
     \def\@secnumfont{\bfseries}
     \makeatother
\newtheorem{theorem}{Theorem}[section]
\newtheorem{lemma}[theorem]{Lemma}
\newtheorem{proposition}[theorem]{Proposition}

\theoremstyle{definition}

\theoremstyle{remark}
\newtheorem{remark}[theorem]{Remark}
\numberwithin{equation}{section}
\usepackage[dvipdfmx]{graphicx}
\usepackage[dvipdfmx]{color}
\usepackage{amsfonts}
\usepackage{amssymb}
\usepackage{mathrsfs}
\usepackage{amsmath,amsthm}
\usepackage{amsmath}

\newcommand{\E}{\mathop {\rm E}}

\begin{document}

\title[Optimal Execution with Uncertain Market Impact]{Theoretical and Numerical Analysis 
of an Optimal Execution Problem with Uncertain Market Impact}

\author{Kensuke Ishitani}
\address{Kensuke Ishitani: Department of Mathematics, 
Faculty of Science and Technology, Meijo University, 
Tempaku, Nagoya 468-8502, Japan}
\email{kishitani@meijo-u.ac.jp}

\author[Takashi Kato]{Takashi Kato*}
\address{Takashi Kato: Division of Mathematical Science for Social Systems, 
Graduate School of Engineering Science, Osaka University, 
1-3 Machikaneyama-cho, Toyonaka, Osaka 560-8531, Japan}
\email{kato@sigmath.es.osaka-u.ac.jp}
\thanks{* This work was supported by a grant-in-aid from the Zengin Foundation for Studies on Economics and Finance. }



\begin{abstract}
This paper is a continuation of \cite {Ishitani-Kato_COSA1}, 
in which we derived a continuous-time value function corresponding to 
an optimal execution problem with uncertain market impact 
as the limit of a discrete-time value function. 
Here, we investigate some properties of the derived value function. 
In particular, we show that the function is continuous and has the semigroup property, 
which is strongly related to the Hamilton--Jacobi--Bellman quasi-variational inequality. 
Moreover, we show that noise in market impact causes risk-neutral assessment to underestimate the impact cost. 
We also study typical examples under a log-linear/quadratic market impact function with Gamma-distributed noise.  
\end{abstract}

\maketitle

\section{Introduction and the Model}\label{sec_intro}

In \cite {Ishitani-Kato_COSA1}, 
we derive a continuous-time value function corresponding to 
an optimal execution problem with uncertain market impact (MI) 
as a limit of a discrete-time value function. 
In this paper, we study some mathematical properties of the value function, 
and give an interpretation from the point of view of mathematical finance. 

First, we recall the continuous-time value function 
derived in \cite {Ishitani-Kato_COSA1}. 
Denote by $\mathcal {C}$ the set of non-decreasing, non-negative, and continuous functions $u$ on 
$D := \Bbb {R}\times [0, \Phi_0] \times [0,\infty )$, with $\Phi _0 > 0$ fixed, such that 
\begin{eqnarray}\label{growth_C}
u(w,\varphi ,s)\leq C_u(1+|w|^{m_u}+s^{m_u}), \ \ (w,\varphi ,s)\in D
\end{eqnarray}
for some constants $C_u, m_u>0$. 
For $t\in [0, 1]$,\ $(w,\varphi ,s)\in D$ and $u\in \mathcal {C}$, define 
\begin{eqnarray}\label{def_conti_v}
V_t(w,\varphi ,s ; u) = \sup _{(\zeta _r)_{r}\in \mathcal {A}_t(\varphi )}
\E [u(W_t,\varphi _t, S_t)] 
\end{eqnarray}
subject to 
\begin{align}\nonumber 
dW_r &= \zeta _rS_rdr, \\\nonumber 
d\varphi _r &= -\zeta _rdr, \\
dX_r &= \sigma (X_r)dB_r+b(X_r)dr-g(\zeta _r)dL_r, 
\label{SDE_X_g}\\\nonumber 
S_r &= \exp (X_r) 
\end{align}
and $(W_0,\varphi _0,S_0) = (w,\varphi ,s)$, 
where $(B_r)_{0\leq r\leq 1}$ is a standard one-dimensional Brownian motion 
defined on a complete probability space $(\Omega , \mathcal {F}, P)$ 
and $(L_r)_{0\leq r\leq 1}$ is a one-dimensional non-decreasing L\'evy process 
(subordinator) defined on the same probability space. 
(Note that $V_0(w,\varphi ,s ; u) = u(w,\varphi ,s)$.)
Assume that $(B_r)_r$ and $(L_r)_r$ are independent. 
Further assume that $\sigma , b : \Bbb {R}\longrightarrow \Bbb {R}$ are Lipschitz continuous bounded functions satisfying
\begin{eqnarray}\label{Bdd_Lipschitz_Constant}
|\sigma (x) - \sigma (y)| + |b(x) - b(y)| \leq K|x - y|, \ \ 
|\sigma (x)| + |b(x)| \leq K, \ \ x, y\in \Bbb {R} 
\end{eqnarray}
for some $K > 0$, 
and $g : [0, \infty )\longrightarrow [0, \infty )$ is a function defined by 
\begin{align*}
g(\zeta)=\int_0^{\zeta} h(\zeta ')d\zeta ' ,
\end{align*}
where $h:[0, \infty) \to [0, \infty)$ is a non-decreasing continuous function. 
$\mathcal {A}_t(\varphi )$ is the set of 
$(\mathcal {F}_r)_{0\leq r\leq t}$-adapted and caglad processes 
(i.e., those that are left-continuous with finite right-limit for arbitrary values of $r$) 
$\zeta  = (\zeta _r)_{0\leq r\leq t}$ 
such that 
$\zeta _r\geq 0$ for each $r\in [0,t]$,\ 
$\int ^t_0\zeta _r dr\leq \varphi $ almost surely, and 
\begin{eqnarray}\label{def_sup}
||\zeta ||_\infty  := \sup _{(r,\omega )\in [0,t]\times \Omega }\zeta _r(\omega ) < \infty , 
\end{eqnarray}
where $\mathcal {F}_r=\sigma \{ B_v, L_v;v\leq r\}\vee \{\mbox{Null sets}\} $. 
Here, the supremum in (\ref {def_sup}) is taken over all values in $[0, t]\times \Omega $. 
As noted in \cite {Ishitani-Kato_COSA1}, 
we may use the essential supremum in (\ref {def_sup}) in place of the supremum. 

We assume that the L\'evy measure $\nu $ of $(L_r)_r$ satisfies 
\begin{eqnarray}\label{assumption_C}
||\nu ||_1 + ||\nu ||_2 < \infty , 
\end{eqnarray}
where $||\nu ||_p = \left( \int _{(0, \infty )}z^p\nu (dz)\right)^{1/p} $. 
Note that the L\'evy decomposition of $(L_r)_r$ is given by
\begin{eqnarray}
L_r = \gamma r + \int ^r_0\int _{(0, \infty )}zN(dv, dz), 
\end{eqnarray}
where $\gamma \geq 0$ and $N(\cdot, \cdot)$ is a Poisson random measure 
(see, for example, \cite {Papapantoleon, Sato}). 

Here, we introduce the financial interpretation of these notations. 
We consider a simple market model 
in which only two financial assets are traded: cash and a security. 
Assume that a single trader is to sell (liquidate) 
the owned shares of the security by time $t$. 
Also assume that the price of the cash is always $1$ 
(in other words, the risk-free rate is $0$) and 
that the security price fluctuates due to market noise and in response to the trader's sales. 
The function $u$ in $\mathcal {C}$ is regarded as the trader's utility function. 
With this, $V_t(w, \varphi , s; u)$ is the supremum of the expected utility of the trader 
with initial cash amount $w$, initial shares $\varphi \in [0, \Phi_0]$, 
and initial security price $s$. 
Here, $\Phi _0 > 0$ denotes an upper bound of $\varphi $ and 
can be arbitrarily chosen; 
$(\zeta _r)_{0\leq r\leq t}$ denotes the trader's execution strategy; and 
$\zeta _r$ denotes the execution speed at time $r$. 
The trader chooses an admissible execution strategy from $\mathcal {A}_t(\varphi )$ 
to optimize the expected utility of the triplet $(W_t, \varphi _t, S_t)$, where 
$S_r$ describes the security price at time $r$ and $X_r$ is its log-price; 
$W_r$ denotes the cash amount at time $r$; and $\varphi_r$ denotes the shares of the security at time $r$. 
The fluctuation of the triplet $(W_r, \varphi _r, S_r)_{0\leq r\leq t}$ 
is characterized by the differential equations in (\ref {SDE_X_g}). 
$(B_r)_r$ represents the component of the market noise reflected in fluctuation of the security price. The term 
\begin{eqnarray}\label{MI_decomp}
g(\zeta _r)dL_r = \gamma g(\zeta _r)dr + g(\zeta _r)\int _{(0, \infty )}zN(dr, dz) 
\end{eqnarray}
describes the (infinitesimal) MI of the trader's selling with speed $\zeta _r$. 
$\gamma $ (resp., $g$) denotes the magnitude (resp., shape) of the MI. 
Because $g$ is non-decreasing and convex, the MI becomes huge 
when $\zeta _r$ is large. The last term in the right-hand side of (\ref {MI_decomp}) 
indicates the effect of noise in the MI, which is mathematically described by the jump of $(L_r)_r$. 

In this paper, we study some properties of the continuous-time value function $V_t(w,\varphi ,s ; u)$. 
We find that the value function is continuous in $(w, \varphi , s)\in D$ and $t > 0$. 
In addition, right-continuity at $t = 0$ depends on the state of 
$h(\infty ) := \lim _{\zeta \rightarrow \infty }h(\zeta )$. 
In particular, noise in the MI does not affect the continuity of the value function. 
We also show that the Bellman principle (the semi-group property) holds 
and perform a comparison with the result in the case of a deterministic MI, which was studied in \cite {Kato}, 
and show that noise in the MI causes risk-neutral assessment to underestimate the MI cost. 
This means that a trader who attempts to minimize the expected liquidation cost 
is not sensitive enough to uncertainty in the MI. 
Last, we present generalizations of the examples from \cite{Kato} 
and investigate the effects of noise in the MI on the optimal strategy of a trader, by numerical experiments. 
We consider a risk-neutral trader's execution problem 
with a log-linear/quadratic MI function with Gamma-distributed noise. 

The rest of this paper is organized as follows. 
In Section \ref {section_Results}, we present our results on the properties of the value function. 
In Section \ref {sec_SO}, we consider the case where the trader must sell all shares of the security, 
which is referred to as the ``sell-off condition.'' 
We also study the optimization problem under the sell-off condition and show that 
the results in \cite [Sect. 4]{Kato} also hold in our model. 
Section \ref {section_comparison} compares deterministic MIs with random (stochastic) MIs in a risk-neutral framework. 
In Section \ref {section_examples}, we present some examples based on the proposed model. 
We conclude this paper in Section \ref {section_conclusion}. 
All proofs are in Section \ref {sec_proof}.

\section{Properties of Value Functions}\label{section_Results}

Regarding the continuity of the continuous-time value function, we have the following theorem: 

\begin{theorem}\label{conti_random} \ Let $u\in \mathcal {C}$.\\
$\mathrm {(i)}$ \ If $h(\infty )=\infty $, then
$V_t(w,\varphi ,s ; u)$ is continuous in $(t,w,\varphi ,s)\in [0,1]\times D$. \\
$\mathrm {(ii)}$ If $h(\infty )<\infty $, then 
$V_t(w,\varphi ,s ; u)$ is continuous in $(t,w,\varphi ,s)\in (0,1]\times D$ and 
$V_t(w,\varphi ,s ; u)$ converges to $Ju(w,\varphi ,s)$ 
uniformly on any compact subset of $D$ as $t\downarrow 0$, where $Ju(w,\varphi ,s)$ is given as 
\begin{eqnarray*}
\left\{
                   \begin{array}{l}
                    	\sup _{\psi \in [0,\varphi ]}
u\Big (w+\frac{1-e^{-\gamma h(\infty )\psi }}{\gamma h(\infty )}s,\varphi -\psi , s e^{-\gamma h(\infty )\psi} \Big ) 
\hspace{5mm} (\gamma h(\infty )>0), \\
                    	\sup _{\psi \in [0,\varphi ]}
u(w+\psi s,\varphi -\psi ,s) \hspace{35.0mm} (\gamma h(\infty )=0). 
                   \end{array}
                   \right. 
\end{eqnarray*}
\end{theorem}

\begin{remark} \ 
\begin{itemize}
 \item [ (i) ] 
The assertions of Theorem \ref {conti_random} are also quite similar to the result in \cite{Kato}, which showed that 
continuities in $w$, $\varphi $, and $s$ of the value function are always guaranteed, but 
continuity in $t$ at the origin depends on the state of the function $h$ at infinity. 
When $h(\infty ) = \infty $, MI for large sales 
is sufficiently strong ($g(\zeta )$ diverges rapidly with $\zeta \rightarrow \infty $) 
to prevent the trader from performing instant liquidation: 
an optimal policy is ``no trading'' in infinitesimal time, 
and thus $V_t$ converges to $u$ as $t\downarrow 0$. 
When $h(\infty ) < \infty $, the value function is not always continuous at $t = 0$ and has the right limit $Ju(w, \varphi ,s)$. 
In this case, MI for large sales is not particularly strong ($g(\zeta )$ still diverges, although with low divergence speed) 
and there is room for liquidation within infinitesimal time. 
The function $Ju(w,\varphi ,s)$ corresponds to the utility of liquidation by the trader, 
who sells part of the shares of a security $\psi $ by dividing it infinitely within an infinitely short time 
(sufficiently short that the fluctuation in the price of the security can be ignored) 
and obtains an amount $\varphi -\psi $; that is, 
\begin{eqnarray}\label{almost_block}
\zeta ^\delta _r = \frac{\psi }{\delta }1_{[0, \delta ]}(r), \ \ r\in [0, t] \ \ (\delta \downarrow 0). 
\end{eqnarray}

Note that, similarly to the argument in Remark 2.6 in \cite {Ishitani-Kato_COSA1}, 
we obtain significant improvement in the strength of the proofs over that given in \cite {Kato}, 
and this is one of the main mathematical contributions of this paper. 
See Section \ref {sec_proof} for details. 
 \item [ (ii) ]
Note that the jump part 
\begin{eqnarray}\label{jump_term}
g(\zeta _r)\int _{(0, \infty )}z N(dr, dz) 
\end{eqnarray}
does not change the result. 
Also note that if $\gamma = 0$ and $h(\infty ) < \infty $, 
then the effect of MI disappears in $Ju(w, \varphi , s)$. 
This situation may occur even if $\E [c^n_k] \geq \varepsilon _0$ 
(or $\E [L_1]\geq \varepsilon _0$) for some $\varepsilon _0 > 0$. 
\end{itemize}
\end{remark}

Here, we present the Bellman principle (dynamic programming principle or ``semi-group'' property). 
Let us define $Q_t:\mathcal {C}\longrightarrow \mathcal {C}$ by 
$Q_t u(w,\varphi,s)=V_t(w,\varphi,s;u)$. 
Then we can easily show that $Q_t$ is well defined as a nonlinear operator. 
The same proof as that for Theorem 3.2 in \cite {Kato} gives the following proposition: 

\begin{proposition}\label{th_semi} \ 
For each $r, t\in [0, 1]$ with $t + r\leq 1$, 
$(w, \varphi , s)\in D$ and $u\in \mathcal {C}$, it holds that 
$Q_{t+r}u(w,\varphi ,s)=Q_t Q_r u(w,\varphi ,s)$. 
\end{proposition}

\begin{remark}
By using the above proposition, 
we can formally derive the Hamilton--Jacobi--Bellman (HJB) equation 
corresponding to our value function on 
the generalized domain of the utility function 
$\hat{D} = \Bbb {R}\times [0, \infty ) \times [0, \infty )$: 
\begin{eqnarray}\label{HJB_V}
\frac{\partial }{\partial t}V_t(w,\varphi ,s ; u) - \sup _{\zeta \geq 0}
\mathscr {L}^\zeta V_t(w,\varphi ,s ; u) = 0 
\end{eqnarray}
with the same boundary conditions as (3.5) in \cite {Kato}, where 
\begin{align*}
\mathscr {L}^\zeta v(t,w,\varphi ,s)
&=
\overline{\mathscr {L}}^{\zeta }v(t,w,\varphi , s) - \tilde{\mathscr {L}}^{\zeta }v(t,w,\varphi , s), \\
\overline{\mathscr {L}}^{\zeta }v(t,w,\varphi , s)
&= 
\frac{1}{2}\hat{\sigma }(s)^2\frac{\partial ^2}{\partial s^2}v(t,w,\varphi ,s) + 
\hat{b}(s)\frac{\partial }{\partial s}v(t,w,\varphi ,s)\\
& + 
\zeta \Big (s\frac{\partial }{\partial w}v(t,w,\varphi ,s) - 
\frac{\partial }{\partial \varphi }v(t,w,\varphi ,s)\Big ) - 
\gamma g(\zeta )s\frac{\partial }{\partial s}v(t,w,\varphi ,s), \\
\tilde{\mathscr {L}}^{\zeta }v(t,w,\varphi , s) 
&= \int _{(0, \infty )}
\left\{ v(w, \varphi , s) - v(w, \varphi , s e^{-g(\zeta )z})\right\} \nu (dz). 
\end{align*} 
(\ref {HJB_V}) is a partial integro-differential equation (PIDE). 
When $\tilde {\mathscr {L}}^{\zeta }\equiv 0$, that is, when there is no jump, 
characterization of our value function 
as the unique viscosity solution of (\ref {HJB_V}) is 
studied by \cite {Kato} under some additional technical conditions. 
Showing these properties in the general case is a more challenging task. 
Here we introduce some related literature in place of 
presenting a detailed argument on the solvability of (\ref {HJB_V}): 
in \cite {Holden}, the existence 
(i.e., characterization of a value function as a viscosity solution) 
and uniqueness of the solution of the HJB equation 
corresponding to the optimal investment/consumption problem with durability and local substitution 
in the L\'evy version of the Black--Scholes-type market model is studied. 
Reference \cite {Seydel} shows existence and uniqueness of a solution to the 
Hamilton--Jacobi--Bellman quasi-variational inequalities (HJBQVIs) 
appearing in combined impulse and (regular) stochastic control problems with jump diffusions 
(existence in this case is also introduced in \cite {Oksendal-Sulem} without detailed technical arguments). 
In \cite {Bouchard-Touzi}, by means of the weak dynamic programming principle, 
the characterization of a value function of stochastic control problems under L\'evy processes 
with finite L\'evy measure, which arises as a discontinuous viscosity solution of the corresponding HJB equation, is studied.  
The strong comparison principle (which is closely related to the uniqueness of viscosity solutions) 
for second-order non-linear PIDEs on a bounded domain is studied in \cite {Ciomaga}. 
\end{remark}

\section{Sell-Off Condition}\label{sec_SO}
In this section, we consider the optimal execution problem under the 
``sell-off condition" introduced in \cite {Kato}. 
A trader has a certain quantity of shares of a security at the initial time, and 
must liquidate all of them by the time horizon. 
Then, the space of admissible strategies is reduced to 
\begin{eqnarray*}
\mathcal {A}^{\mathrm {SO}}_t(\varphi ) = 
\left \{ (\zeta _r)_r\in \mathcal {A}_t(\varphi )\ ; \ 
\int ^t_0\zeta _r dr = \varphi \right \} . 
\end{eqnarray*}
We define a value function with the sell-off condition by 
\begin{eqnarray*} 
V^{\mathrm {SO}}_t(w,\varphi ,s ; U) &=& \sup _{(\zeta _r)_r\in \mathcal {A}^{\mathrm {SO}}_t(\varphi )}\E [U(W_t)] 
\end{eqnarray*}
for a continuous, non-decreasing and polynomial growth function $U : \Bbb {R}\longrightarrow \Bbb {R}$. 

The following theorem is analogous to Theorem 4.1 in \cite {Kato} 
(we omit the proof because it is nearly identical): 

\begin{theorem} \ \label{thesame}$V^{\mathrm {SO}}_t(w,\varphi ,s ; U) = V_t(w,\varphi ,s ; u)$, 
where $u(w, \varphi , s) = U(w)$. 
\end{theorem}

By Theorem \ref {thesame}, we see that the sell-off condition does not introduce changes 
in the value of the value function in a continuous-time model. 

Analogously to Theorem 4.2 in \cite {Kato}, 
a similar result to Theorem 3 in \cite {Lions-Lasry} holds when $g(\zeta )$ is linear: 
\begin{theorem} \ \label{th_LL}Assume $g(\zeta ) = \alpha _0\zeta $ for $\alpha _0> 0$. \\
$\mathrm {(i)}$ \ $V^\mathrm {SO}_t(w, \varphi , s ; U) = \overline{V}^\varphi _t
\left( w + \frac{1 - e^{-\gamma \alpha _0 \varphi }}{\gamma \alpha _0}s, e^{-\gamma \alpha _0 \varphi }s ; U\right)$, where 
\begin{eqnarray*}
\overline{V}^\varphi _t(\bar{w}, \bar{s} ; U) &=& \sup _{(\overline{\varphi }_r)_r\in \overline {\mathcal {A}}_t(\varphi )}
\E [U(\overline{W}_t)]\\
&&\hspace{1mm}\mathrm {s.t.}\hspace{3.3mm}d\overline{S}_r = 
e^{-\gamma \alpha _0\overline{\varphi }_r}\hat{b}(\overline{S}_re^{\gamma \alpha _0\overline{\varphi }_r})dr + 
e^{-\gamma \alpha _0 \overline{\varphi }_r}\hat{\sigma }(\overline{S}_re^{\gamma \alpha _0 \overline{\varphi }_r})dB_r\\
&&\hspace{19mm} -\overline{S}_{r-}dG_r, \\
&&\hspace{8mm}d\overline{W}_r = \frac{e^{\gamma \alpha _0\overline{\varphi }_r} - 1}{\gamma \alpha _0}d\overline{S}_r, \\
&&\hspace{11mm}\overline{S}_0 = \bar{s}, \ \ \overline{W}_0 = \bar{w} 
\end{eqnarray*}
and
\begin{eqnarray*}
\overline {\mathcal {A}}_t(\varphi ) &=& \left\{ \left( \varphi - \int ^r_0\zeta _vdv\right) _{0\leq r\leq t}\ ; \ 
(\zeta _r)_{0\leq r\leq t} \in \mathcal {A}^\mathrm {SO}_t(\varphi )\right\} , \\
G_r &=& \int_0^r \int_{(0, \infty)}(1-e^{-\alpha _0\zeta _sz})N(ds, dz). 
\end{eqnarray*}
$\mathrm {(ii)}$ \ If $U$ is concave and $\hat{b} (s)\leq 0$ for $s\geq 0$, then 
\begin{eqnarray}\label{eq_LL}
V^\mathrm {SO}_t(w, \varphi , s ; U) = 
U\left( w + \frac{1 - e^{-\gamma \alpha _0 \varphi }}{\gamma \alpha _0 }s\right) . 
\end{eqnarray}
\end{theorem}

The proof is in Section \ref {sec_proof_th_LL}. 
Note that the assertion (ii) is the same as Theorem 3 in \cite {Lions-Lasry}, 
and in this case we can also obtain the explicit form of the value function. 
The right side of (\ref {eq_LL}) is equal to $Ju(w,\varphi ,s)$ for $u(w, \varphi , s) = U(w)$ 
and the nearly optimal strategy for $V^{\mathrm {SO}}_t(w,\varphi ,s ; U) = V_t(w,\varphi ,s ; u)$ 
is given by (\ref {almost_block}). 
This implies that when considering a linear MI function, 
a risk-averse (or risk-neutral) trader's optimal liquidation strategy with negative risk-adjusted drift 
is nearly the same as block liquidation (i.e., selling all shares at once) at the initial time.

\section{Effect of Uncertainty in MI in the Risk-neutral Framework}\label{section_comparison}

The purpose of this section is to investigate how noise in the MI function affects the trader. 
Particularly, we focus on the case where the trader is risk-neutral, that is, 
$u(w, \varphi , s) = u_{\mathrm {RN}}(w, \varphi , s) = w$. Note that such a risk-neutral setting 
is a typical and standard assumption in the study of the execution problem 
(see e.g. \cite{Alfonsi-Fruth-Schied,Cheng-Wang,Kato2,Kato_VWAP,Konishi-Makimoto,Makimoto-Sugihara,Schied-Zhang}).

First, we prepare a value function of the execution problem with a deterministic MI function to 
compare with the case of random MI. 
Let $\bar{V}_t(w, \varphi , s ; u)$ be the same as in (\ref {def_conti_v}) 
by replacing $g(\zeta )$ and $L_t$ with $\tilde{\gamma }g(\zeta )$ and $t$, that is, 
the SDE for $(X_r)_r$ is given as 
\begin{eqnarray*}
dX_r = \sigma (X_r)dB_r + b(X_r)dr - \tilde{\gamma }g(\zeta _r)dr, 
\end{eqnarray*}
where 
\begin{eqnarray}\label{def_tilde_gamma}
\tilde{\gamma } = \E [L_1] = \gamma + \int _{(0, \infty )}z\nu (dz). 
\end{eqnarray}
The following proposition is proved in Section \ref {sec_proof_th_comp_noise}: 

\begin{proposition} \ \label{th_comp_noise}We have
\begin{eqnarray}\label{comp_noise}
V_t(w, \varphi , s ; u_\mathrm {RN}) \geq \bar{V}_t(w, \varphi , s ; u_\mathrm {RN}). 
\end{eqnarray}
\end{proposition}

This proposition shows that 
noise in MI is welcome because it decreases the liquidation cost for a risk-neutral trader. 

For instance, we consider a situation where 
the trader estimates the MI function from historical data and 
tries to minimize the expected liquidation cost. 
Then, a higher sensitivity of the trader to the volatility risk of MI 
results in a lower estimate for the expected proceeds of the liquidation. 
This implies that accommodating the uncertainty in MI makes the trader prone to
underestimating the liquidation cost. 
Thus, as long as the trader's target is the expected cost, 
the uncertainty in MI is not an incentive for being conservative 
with respect to the unpredictable liquidity risk. 
In Section \ref {section_examples}, we present the results of numerical experiments 
conducted to simulate the above phenomenon. 

\section{Examples}\label{section_examples}

In this section, we show two examples of our model, 
which are both generalizations of the ones in \cite {Kato}. 

Motivated by the Black--Scholes-type market model, 
we assume that $b(x) \equiv -\mu $ and $\sigma (x) \equiv \sigma $ for some 
constants $\mu , \sigma \geq 0$ and assume that $\tilde{\mu } := \mu - \sigma ^2/2$ is positive. 
We also assume a risk-neutral trader with utility function $u(w, \varphi ,s) = u_\mathrm {RN}(w) = w$. 
In this case, if there is no MI, then a risk-neutral trader will fear a decrease in the expected stock price, 
and thus will liquidate all the shares immediately at the initial time. 

We consider MI functions that are log-linear and log-quadratic with respect to liquidation speed, 
and assume Gamma-distributed noise; 
that is, $g(\zeta ) = \alpha _0\zeta ^p$ for $\alpha _0 > 0$ and $p = 1, 2$, and $L_t$ satisfies 
\begin{eqnarray*}
P(L_t - \gamma t\in dx) &=& \mathrm {Gamma}(\alpha _1t, \beta _1)(dx) \\
&:=& 
\frac{1}{\Gamma (\alpha _1t)(\beta _1)^{\alpha _1t}}x^{\alpha _1t - 1}e^{-x/\beta _1} 1_{(0,\infty )}(x)\,dx, 
\end{eqnarray*}
where $\Gamma (x)$ is the Gamma function. 
Here, $\alpha _1, \beta _1$, and $\gamma > 0$ are constants. The corresponding L\'evy measure is 
\begin{eqnarray*}
\nu (dz) = \frac{\alpha _1}{z}e^{-z/\beta _1}1_{(0, \infty )}(z)\,dz. 
\end{eqnarray*}
Note that for the discrete-time model studied in \cite {Ishitani-Kato_COSA1}, 
we can define the corresponding discrete-time MI function as 
$g^n_k(\psi ) = c^n_kg_n(\psi )$, where $g_n(\psi ) = n^{p-1}\alpha _0\psi ^p$ and 
$(c^n_k)_k$ is a sequence of i.i.d.~random variables with distribution 
\begin{eqnarray*}
P(c^n_k - \gamma \in dx) = \mathrm {Gamma}(\alpha _1/n, n\beta _1)(dx). 
\end{eqnarray*}
In each case, assumptions [A], [B1]--[B3], and [C] of \cite {Ishitani-Kato_COSA1} are satisfied.

\subsection{Log-Linear Impact \& Gamma Distribution}\label{sec_linear_eg}

In this subsection, we set $g(\zeta ) = \alpha _0\zeta $ ($p=1$). 
Theorem \ref {th_LL} directly implies the following: 

\begin{theorem}\label{th_eg_random} \ 
We have
\begin{eqnarray}\label{eg_Ju}
V_t(w,\varphi ,s ; u_\mathrm {RN}) = w+\frac{1-e^{-\gamma \alpha _0\varphi }}{\gamma \alpha _0}s 
\end{eqnarray}
for each $t\in (0,1]$ and $(w,\varphi ,s)\in D$. 
\end{theorem}

The implication of this result is the same as in \cite {Kato}: 
the right side of (\ref {eg_Ju}) is equal to $Ju(w,\varphi ,s)$ and converges to 
$w+\varphi s$ as $\alpha _0\downarrow 0$ or $\gamma \downarrow 0$, which is the profit 
gained by choosing the execution strategy of block liquidation at $t = 0$. 
Therefore, the optimal strategy in this case is 
to liquidate all shares by dividing infinitely within an infinitely short time at $t = 0$ 
(we refer to such a strategy as a nearly block liquidation at the initial time). 
Note that the jump part of MI (\ref {jump_term}) 
does not influence the value of $V_t(w,\varphi ,s ; u_\mathrm {RN})$. 

\subsection{Log-Quadratic Impact \& Gamma Distribution}\label{sec_quad_eg}

Next we study the case of $g(\zeta ) = \alpha _0\zeta ^2$ ($p=2$). 
In \cite {Kato}, we obtained a partial analytical solution to the problem: 
when $\varphi $ is sufficiently small or large, we obtain the explicit form of optimal strategies. 
However, the noise in MI complicates the problem, and deriving the explicit solution is more difficult. 
Thus, we rely on numerical simulations. 
Under the assumption that the trader is risk-neutral, 
we can assume that an optimal strategy is deterministic. 
Here, we introduce the following additional condition: \vspace{2mm}\\
$[D]$ \ $\gamma \geq  \alpha _1\beta _1 / 8$. \vspace{2mm}\\
In fact, we can replace our optimization problem with the deterministic control problem 
\begin{eqnarray*}
f(t, \varphi ) = \sup _{(\zeta _r)_r}\int ^t_0
\exp \left( -\int ^r_0q(\zeta _v)dv \right) \zeta _rdr 
\end{eqnarray*}
for a deterministic process $(\zeta _r)_r$ under the above assumption, where 
\begin{align*}
q(\zeta ) &= \tilde {\mu } + \hat{g}(\zeta ), \\
\hat{g}(\zeta ) &= \gamma \alpha _0\zeta ^2 + \alpha _1\log (\alpha _0\beta _1 \zeta ^2 + 1). 
\end{align*}
This gives the following theorem: 
\begin{theorem}\label{th_f} \ 
$V_t(w, \varphi , s ; u_\mathrm {RN}) = w + sf(t, \varphi )$ under $[D]$. 
\end{theorem}

This theorem is obtained by a similar proof to Proposition 5.1 in \cite {Kato} 
by using the following Laplace transform of the Gamma distribution: 
\begin{eqnarray*}
\E[e^{- \lambda c^n_k}] = 
\exp \left( -\gamma \lambda  - \frac{ \alpha_{1}}{n}\log ( n\beta_{1}\lambda  + 1) \right). 
\end{eqnarray*}

From Theorem \ref {th_f} and (\ref {HJB_V}), 
we derive the HJB equation for the function $f$ as 
\begin{eqnarray}\label{HJB_f_eg}
\frac{\partial }{\partial t}f + \tilde{\mu }f - \sup _{\zeta \geq 0}
\left\{ \zeta \left( 1 - \frac{\partial }{\partial \varphi }f\right) - 
\hat{g}(\zeta )f\right\} = 0 
\end{eqnarray}
with the boundary condition 
\begin{eqnarray}\label{boundary_f}
f(0, \varphi ) = f(t, 0) = 0. 
\end{eqnarray}
When $\gamma \geq \alpha _1 / 2$, the function $\hat{g}$ becomes convex, so 
we can apply Theorems 3.3 and 3.6 in \cite {Kato} to show the following proposition: 
\begin{proposition} \ \label{prop_eg_HJB} 
Assume $\gamma \geq \alpha _1/2$. 
Then $f(t, \varphi )$ is the viscosity solution of $(\ref {HJB_f_eg})$. 
Moreover, if $\tilde{f}$ is a viscosity solution of 
$(\ref {HJB_f_eg})$ and $(\ref {boundary_f})$ and 
has a polynomial growth rate, then $f = \tilde{f}$. 
\end{proposition}

It is difficult to obtain an explicit form of the solution of (\ref {HJB_f_eg}) and (\ref {boundary_f}). 
Instead, we solve this problem numerically by considering 
the deterministic control problem $f^n_{[nt]}(\varphi )$ 
in the discrete-time model for a sufficiently large $n$: 
\begin{align*}
f^n_k(\varphi ) &= 
\sup _{
\substack{(\psi ^n_l)^{k-1}_{l = 0}\subset [0, \varphi ]^k,\\
\sum _l\psi ^n_l \leq \varphi  }
}
\sum ^{k-1}_{l = 0}\psi ^n_l\exp \left( -\tilde{\mu }\times \frac{l}{n} - \sum ^l_{m = 0}
I_m\right), \\
I_m &= n\gamma \alpha _0(\psi ^n_m)^2 
+ \frac{\alpha _1}{n}\log (n^2\alpha _0\beta _1 (\psi ^n_m)^2 + 1). 
\end{align*}
Note that the convergence 
$\lim_{n\to \infty} f^n_{[nt]}(\varphi ) = f(t,\varphi )$ 
is guaranteed by Theorem 2.3 of \cite {Ishitani-Kato_COSA1}.
We set each parameter as follows: $\alpha _0 = 0.01, t = 1, 
\tilde{\mu } = 0.05, w = 0, s = 1$, and $n = 500$. 
We examine three patterns for $\varphi$, 
$\varphi  = 1, 10$, and $100$.

\subsubsection{The case of fixed $\gamma $}\label{sec_fixed_gamma}

In this subsection, we set $\gamma = 1$ 
to examine the effects of the shape parameter $\alpha _1$ of the noise in MI. 
Here, we also set $\beta _1 = 2$. As seen in the numerical experiment in \cite {Kato}, 
the forms of optimal strategies vary according to the value of $\varphi$. 
Therefore, we summarize our results separately for each $\varphi$. 

Figure \ref {graph_phi1_0} shows graphs of the optimal strategy $(\zeta _r)_r$ and 
its corresponding process $(\varphi _r)_r$ of the security holdings 
in the case of $\varphi=1$, that is, the number of initial shares of the security is small. 
As found in \cite {Kato}, if there is no noise in the MI function (i.e., if $\alpha _1 = 0$), 
then the optimal strategy is to sell the entire amount at the same speed 
(note that the roundness at the corner in the left graph of Figure \ref {graph_phi1_0} represents 
the discretization error and is not essential). 
The same tendency is found in the case of $\alpha _1 = 1$, 
but in this case the execution time is longer than in the case of $\alpha _1 = 0$. 
When we take $\alpha _1 = 3$, the situation is completely different. 
In this case, the optimal strategy is to increase the execution speed as the time horizon approaches. 

When the amount of the security holdings is $10$, which is larger than in the case of $\varphi=1$, 
the optimal strategy and the corresponding process of the security holdings are as shown in Figure \ref {graph_phi10_0}. 
In this case, a trader's optimal strategy is to increase the execution speed as 
the end of the trading time approaches, which is the same as in the case of $\varphi=1$ with $\alpha _1 = 3$. 
Clearly, a larger value of $\alpha _1$ corresponds to a higher speed of execution closer to the time horizon. 
We should add that a trader cannot complete the liquidation when $\alpha _1 = 3$. 
However, as mentioned in Section \ref {sec_SO}, 
we can choose a nearly optimal strategy from $\mathcal {A}^\mathrm {SO}_1(\varphi )$ 
without changing the value of the expected proceeds of liquidation by 
combining the execution strategy in Figure \ref {graph_phi10_0} (with $\alpha _1 = 3$) 
and the terminal (nearly) block liquidation. See Section 5.2 of \cite {Kato} for details. 

When the amount of the security holdings is too large, as in the case of $\varphi = 100$, 
a trader cannot complete the liquidation regardless of the value of $\alpha _1$, 
as Figure \ref {graph_phi100_0} shows. This is similar to the case of $\varphi = 10$ with $\alpha _1 = 3$. 
The remaining amount of shares of the security at the time horizon is larger for larger noise in MI. 
Note that the trader can also sell all the shares of the security without decreasing the profit 
by combining the strategy with the terminal (nearly) block liquidation.

\begin{figure}[htbp]
\begin{center}
\includegraphics[scale=0.3]{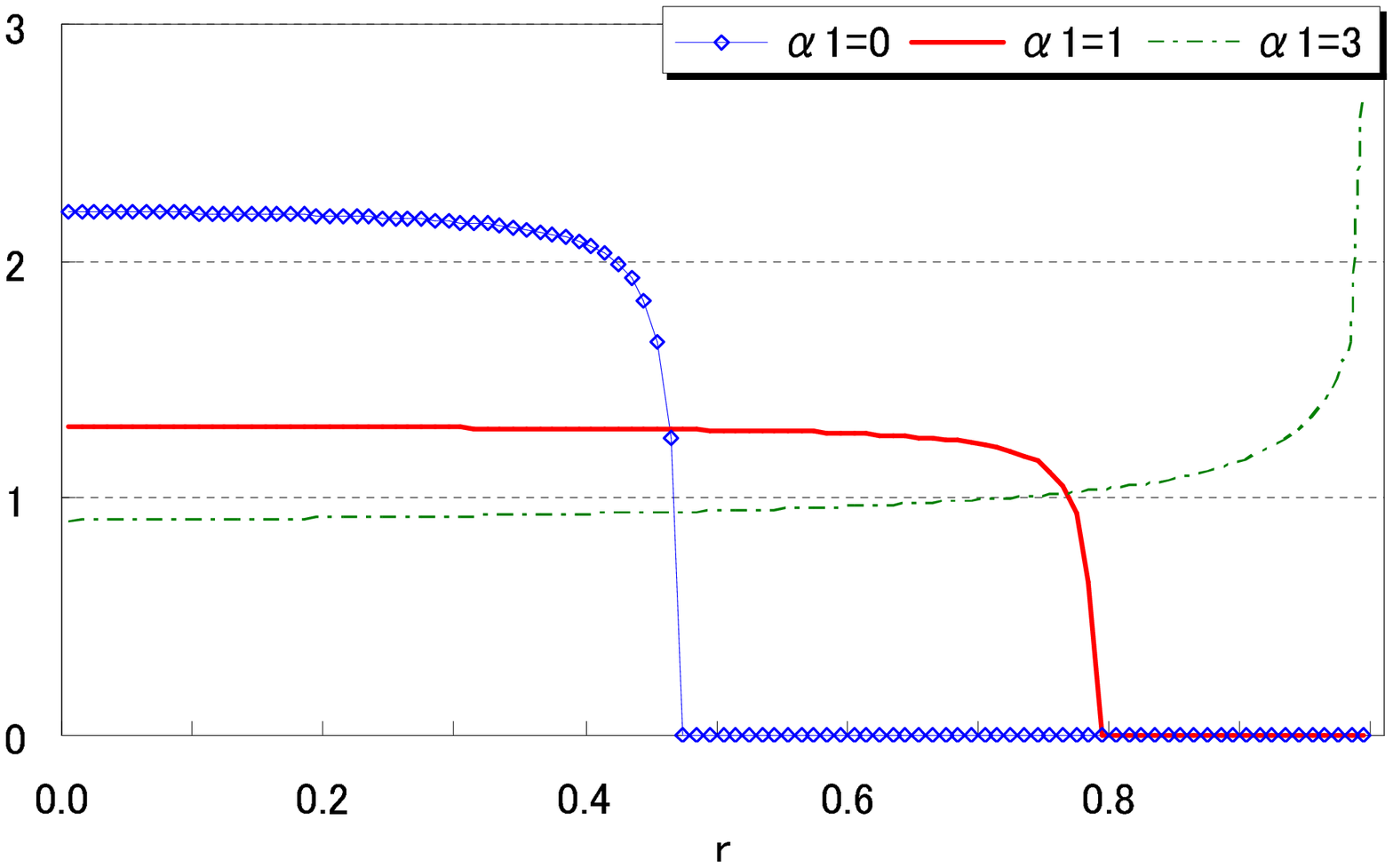}
\includegraphics[scale=0.3]{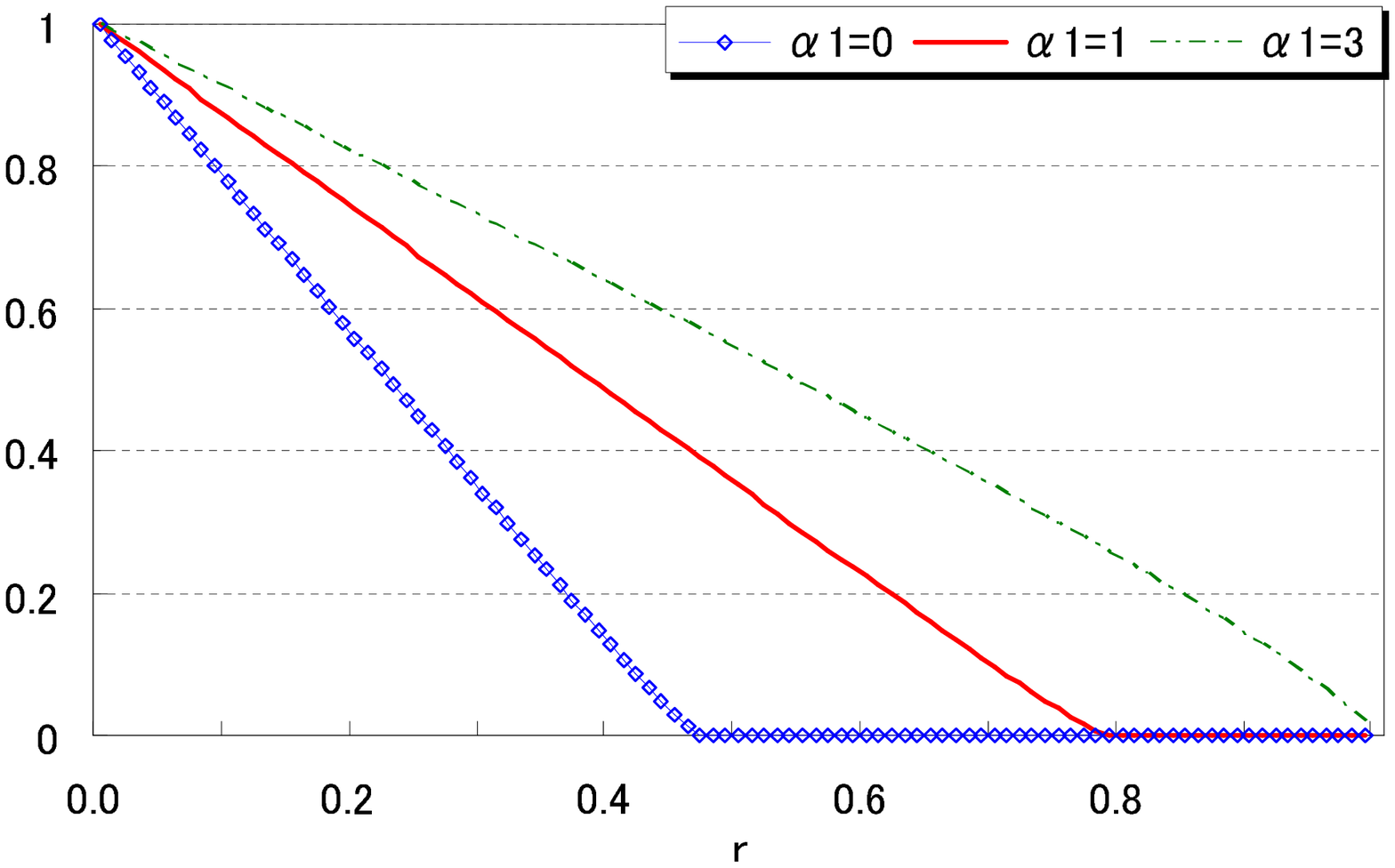}
\caption{Result for $\varphi = 1$ in the case of fixed $\gamma $. 
Left: The optimal strategy $\zeta _r$. 
Right: The amount of security holdings $\varphi _r$. 
}
\label{graph_phi1_0}
\includegraphics[scale=0.3]{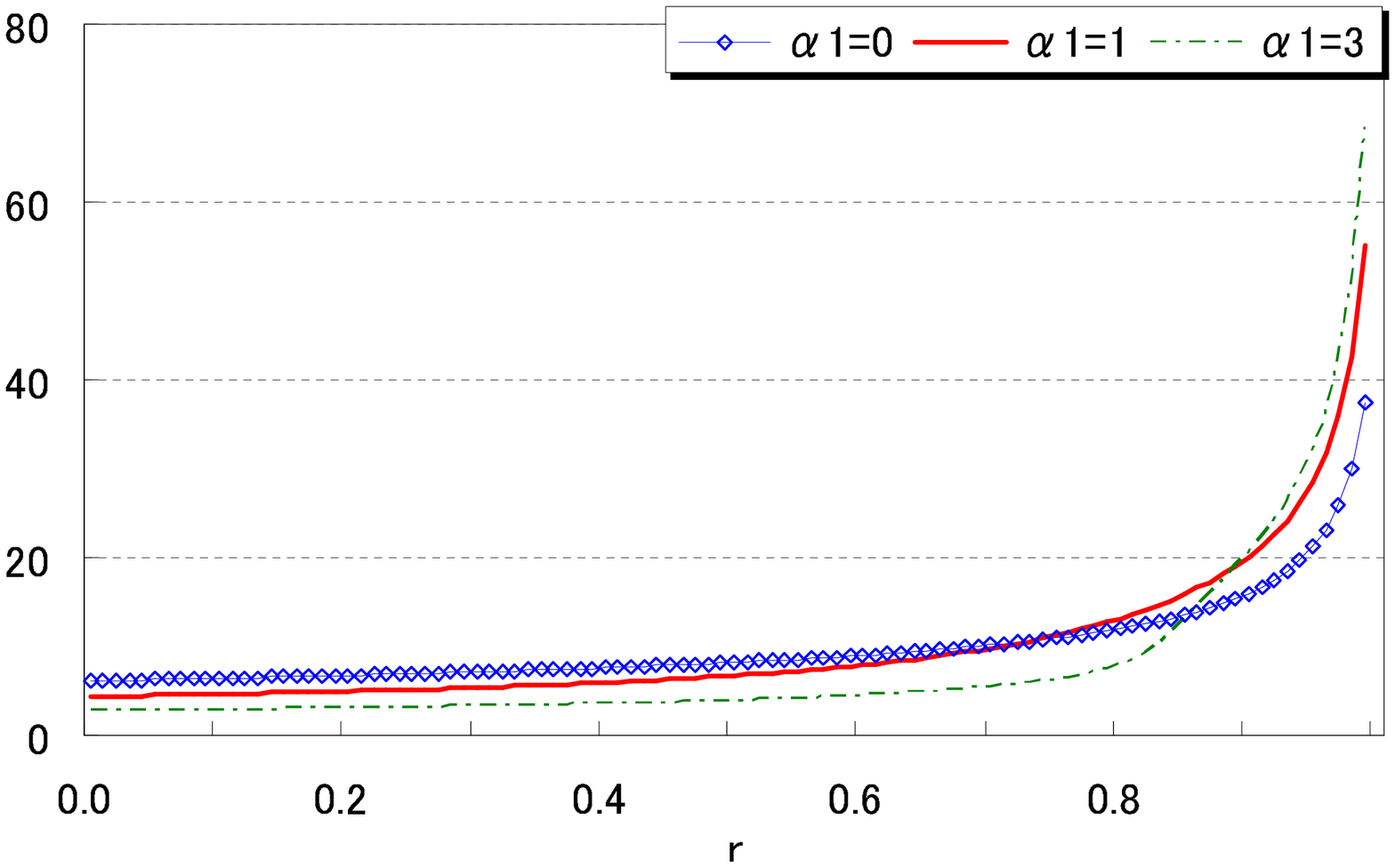}
\includegraphics[scale=0.3]{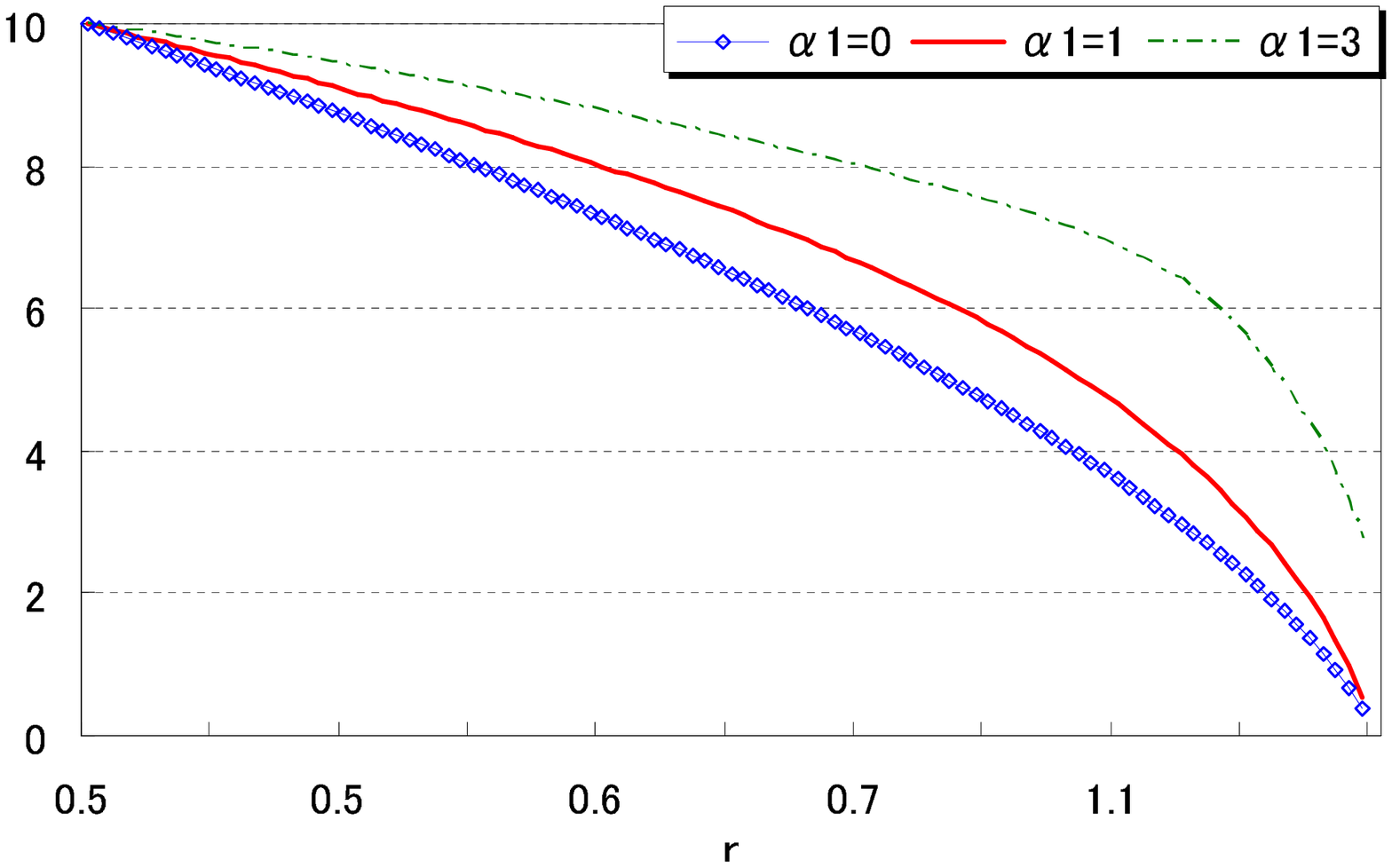}
\caption{Result for $\varphi = 10$ in the case of fixed $\gamma $. 
Left : The optimal strategy $\zeta _r$. 
Right : The amount of security holdings $\varphi _r$. 
}
\label{graph_phi10_0}
\includegraphics[scale=0.3]{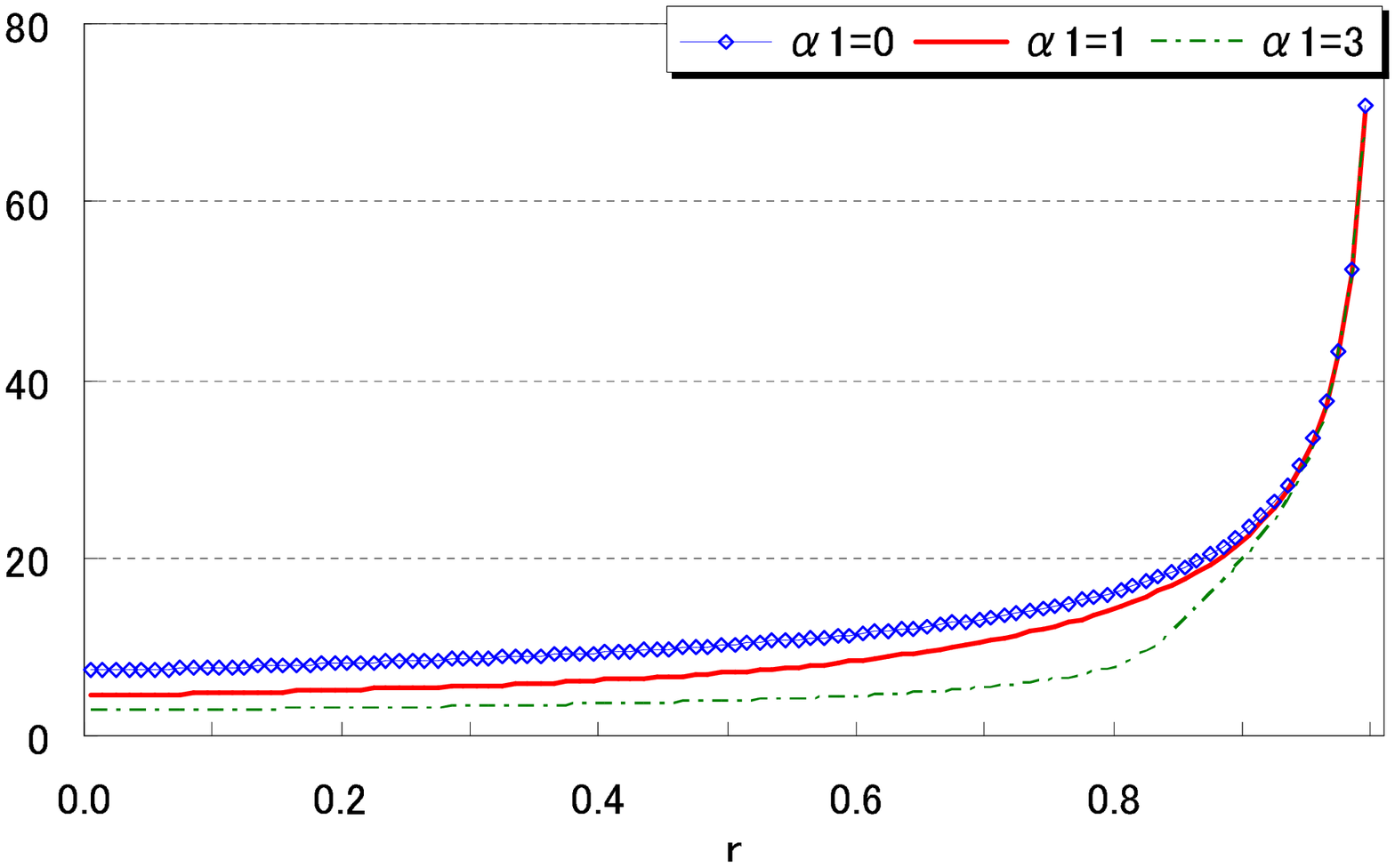}
\includegraphics[scale=0.3]{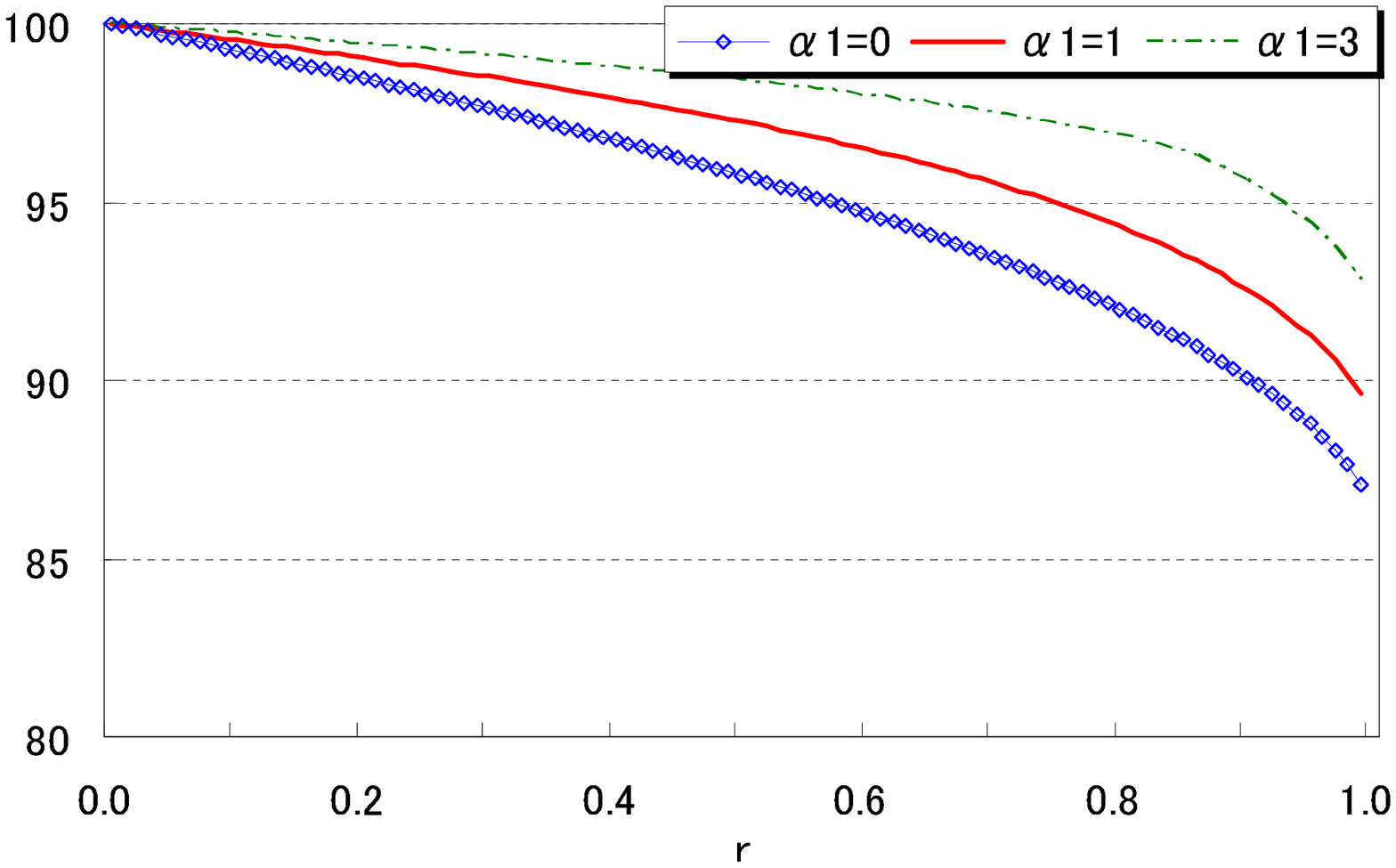}
\caption{Result for $\varphi = 100$ in the case of fixed $\gamma $. 
Left : The optimal strategy $\zeta _r$. 
Right : The amount of security holdings $\varphi _r$. 
}
\label{graph_phi100_0}
\end{center}
\end{figure}

\subsubsection{The case of fixed $\tilde{\gamma }$}

In the above subsection, we presented a numerical experiment performed to 
compare the effects of the parameter $\alpha _1$ by fixing $\gamma $. 
Here, we perform numerical comparison from a different viewpoint. 

The results in Section \ref {section_comparison} imply that accounting for the uncertainty in MI 
will cause a risk-neutral trader to be optimistic about the estimation of liquidity risks. 
To obtain a deeper insight, we investigate the structure of the MI function in more detail. 
In Theorems \ref {conti_random}(ii) and \ref {th_eg_random}, 
the important parameter is $\gamma $, which is the infimum of $L_1$ and is smaller than or equal to $\E [L_1]$. 
We can interpret this as a characteristic feature 
whereby the (nearly) block liquidation eliminates the effect of positive jumps of $(L_t)_t$. 
However, there is another decomposition of $L_t$ such that 
\begin{eqnarray*}
L_t = \tilde{\gamma }t + \int ^t_0\int _{(0, \infty )}z\tilde{N}(dr, dz), 
\end{eqnarray*}
where $\tilde{\gamma }$ is given by (\ref {def_tilde_gamma}) and 
$$\tilde{N}(dr, dz) = N(dr, dz) - \nu (dz)dr.$$ 
This representation is essential from the viewpoint of martingale theory. 
Here, $\tilde{N}(\cdot, \cdot)$ is the compensator of $N(\cdot, \cdot)$ 
and $\tilde{\gamma }$ can be regarded as the ``expectation'' of the noise in MI. 
Just for a risk-neutral world (in which a trader is risk-neutral), 
as studied in Section \ref {section_comparison}, we can compare our model with 
the case of deterministic MI functions as in \cite {Kato} by setting $\tilde{\gamma } = 1$. 
Based on this, we conduct another numerical experiment with 
a constant value of $\tilde{\gamma }$. 

Note that in our example 
\begin{eqnarray}\label{temp_mean}
\tilde{\gamma } = \gamma + \alpha _1\beta _1
\end{eqnarray}
and 
\begin{eqnarray}\label{temp_var}
\frac{1}{t}\mathrm {Var}\left( \int ^t_0\int _{(0, \infty )}z\tilde{N}(dr, dz)\right) = \alpha _1\beta ^2_1 
\end{eqnarray}
hold. Here, (\ref {temp_mean}) (respectively, (\ref {temp_var})) corresponds to the mean 
(respectively, the variance) of the noise in the MI function at unit time. 
Comparisons in this subsection are performed with the following assumptions: 
We set the parameters $\beta _1$ and $\gamma $ to satisfy 
\begin{eqnarray*}
\gamma + \alpha _1\beta _1= 1, \ \ \alpha _1\beta _1^2= 0.5. 
\end{eqnarray*}
We examine the cases of $\alpha _1 = 0.5$ and $1$, 
and compare them with the case of $\gamma = 1$ and $\alpha _1 = 0$. 

Figure \ref{graph_phi1} shows the case of $\varphi = 1$, 
where the trader has a small amount of security holdings. 
Compared with the case in Section \ref {sec_fixed_gamma}, 
the forms of all optimal strategies are the same; that is, 
the trader should sell the entire amount at the same speed. 
The execution times for $\alpha _1 > 0$ are somewhat shorter than for $\alpha _1 = 0$. 

Figure \ref {graph_phi10} corresponds to the case of $\varphi = 10$. 
The forms of the optimal strategies are similar to the case of 
$\varphi = 10$, $\alpha _1 = 0, 1$ in Section \ref {sec_fixed_gamma}. 
Clearly, the speed of execution near the time horizon increases with increasing $\alpha _1$. 

The results for $\varphi = 100$ are shown in Figure \ref {graph_phi100}. 
The forms of the optimal strategies are similar to the case of $\varphi = 100$ in Section \ref {sec_fixed_gamma}. 
However, in contrast to the results in the previous subsection, 
the remaining amount of shares of the security at the time horizon is smaller for larger $\alpha _1$. 

Finally, we investigate the total MI cost introduced in \cite {Kato2} 
(which is essentially equivalent to an implementation shortfall (IS) cost 
\cite {Almgren-Chriss, Perold}): 
\begin{eqnarray*}
\mathrm {TC}(\varphi) = -\log \frac{V_T(0, \varphi, s)}{\varphi s}. 
\end{eqnarray*}
As noted at the beginning of this section, 
when the market is fully liquid and there is no MI, 
then the total proceeds of liquidating $\varphi$ shares of the security at $t=0$ are equal to $\varphi s$. 
In the presence of MI, however, the optimal total proceeds decrease to 
$V_T(0, \varphi, s) = \varphi s\times \exp (-\mathrm {TC}(\varphi))$. 
Thus, the total MI cost $\mathrm {TC}(\varphi)$ denotes the loss rate caused by MI in a risk-neutral world. 

Figure \ref{graph_MIcost} shows the total MI costs in the cases of $\varphi = 1$ and $10$. 
Here, we omit the case of $\varphi = 100$ because the amount of shares of the security 
is too large to complete the liquidation unless otherwise combining terminal block liquidations 
(which may crash the market). In both cases of $\varphi= 1$ and $10$, 
we find that the total MI cost decreases by increasing $\alpha _1$. 
Since the expected value $\tilde{\gamma }$ of the noise in MI is fixed, 
an increase in $\alpha _1$ implies a decrease in $\gamma $ and $\beta _1$. 
Risk-neutral traders seem to be more sensitive to the parameter $\gamma $ than to $\alpha _1$, 
and thus the trader can liquidate the security without concern about the volatility of the noise in MI. 
Therefore, the total MI cost for $\alpha _1 > 0$ is lower than that for $\alpha _1 = 0$. 

\begin{figure}[htbp]
\begin{center}
\includegraphics[scale=0.3]{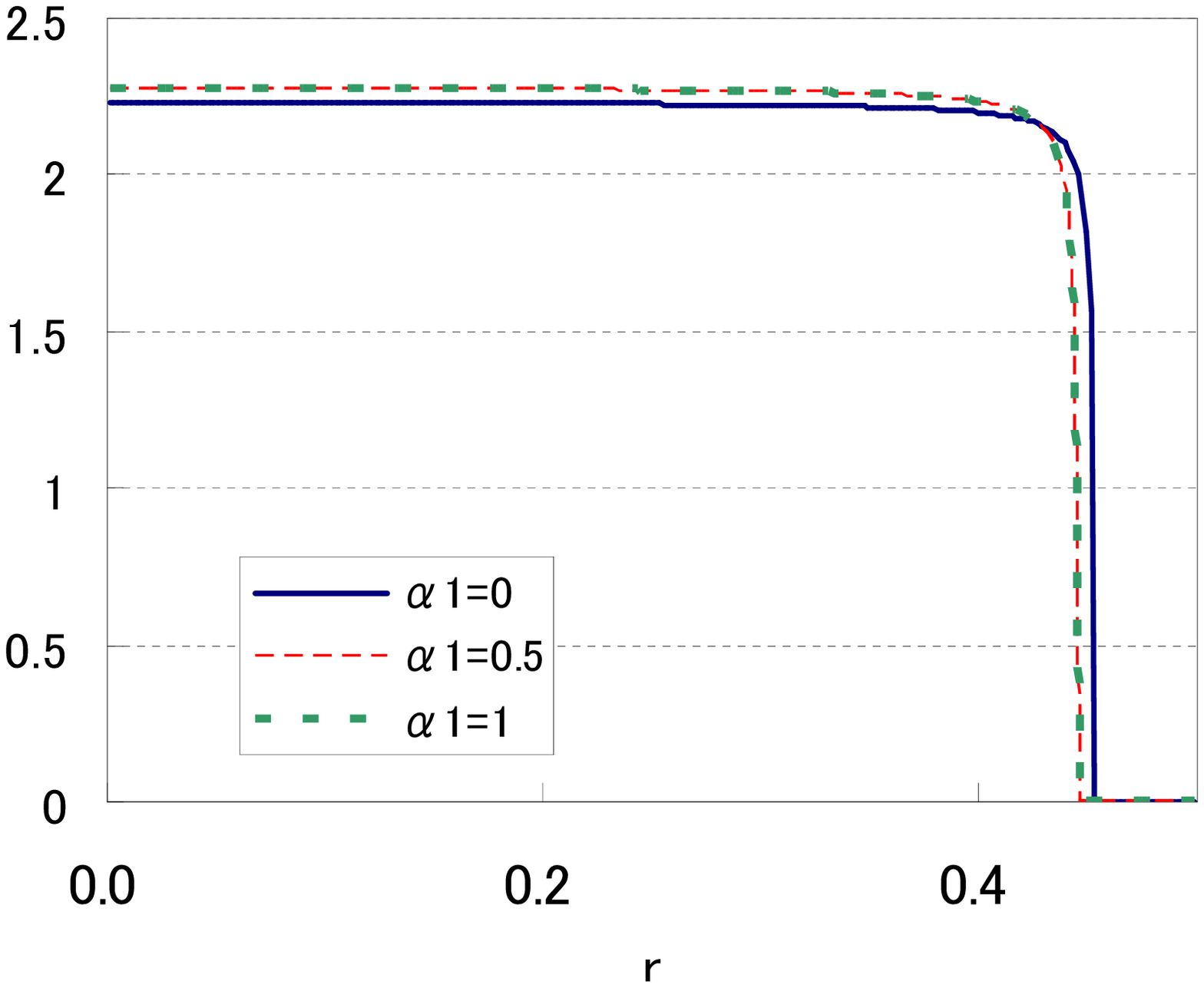}
\includegraphics[scale=0.3]{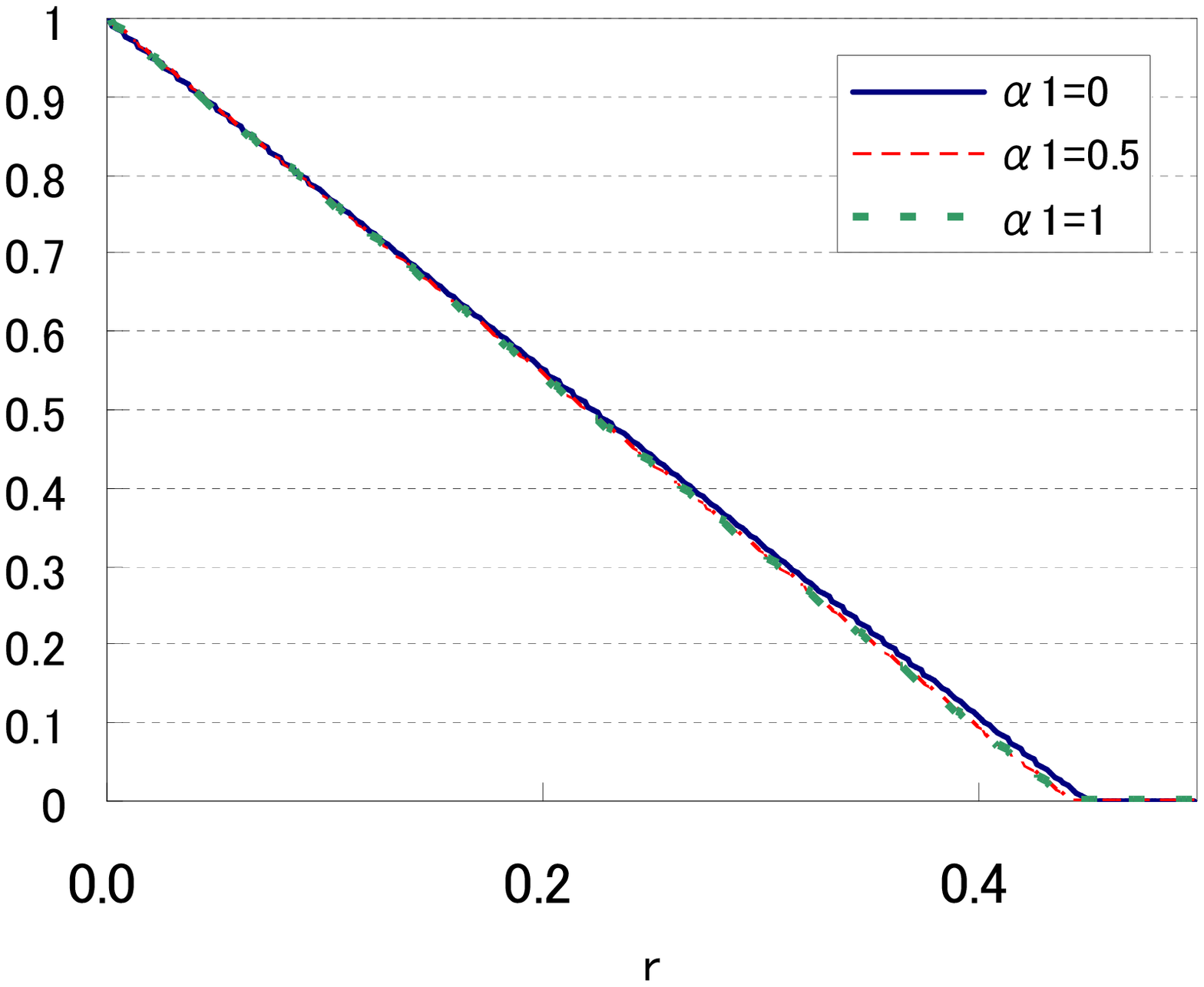}
\caption{Result for $\varphi = 1$ in the case of fixed $\tilde{\gamma }$. 
Left : The optimal strategy $\zeta _r$. 
Right : The amount of security holdings $\varphi _r$. 
}
\label{graph_phi1}
\includegraphics[scale=0.3]{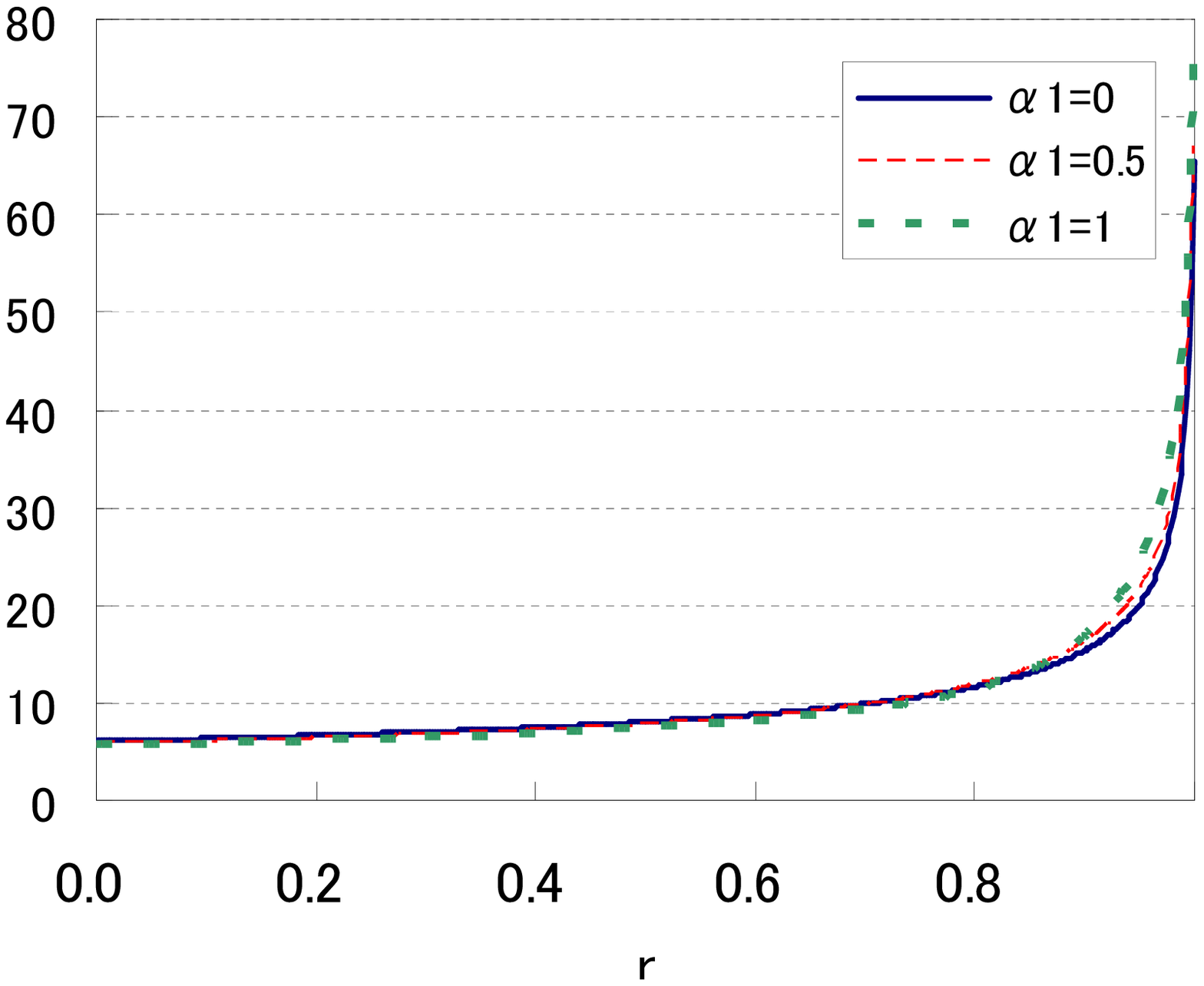}
\includegraphics[scale=0.3]{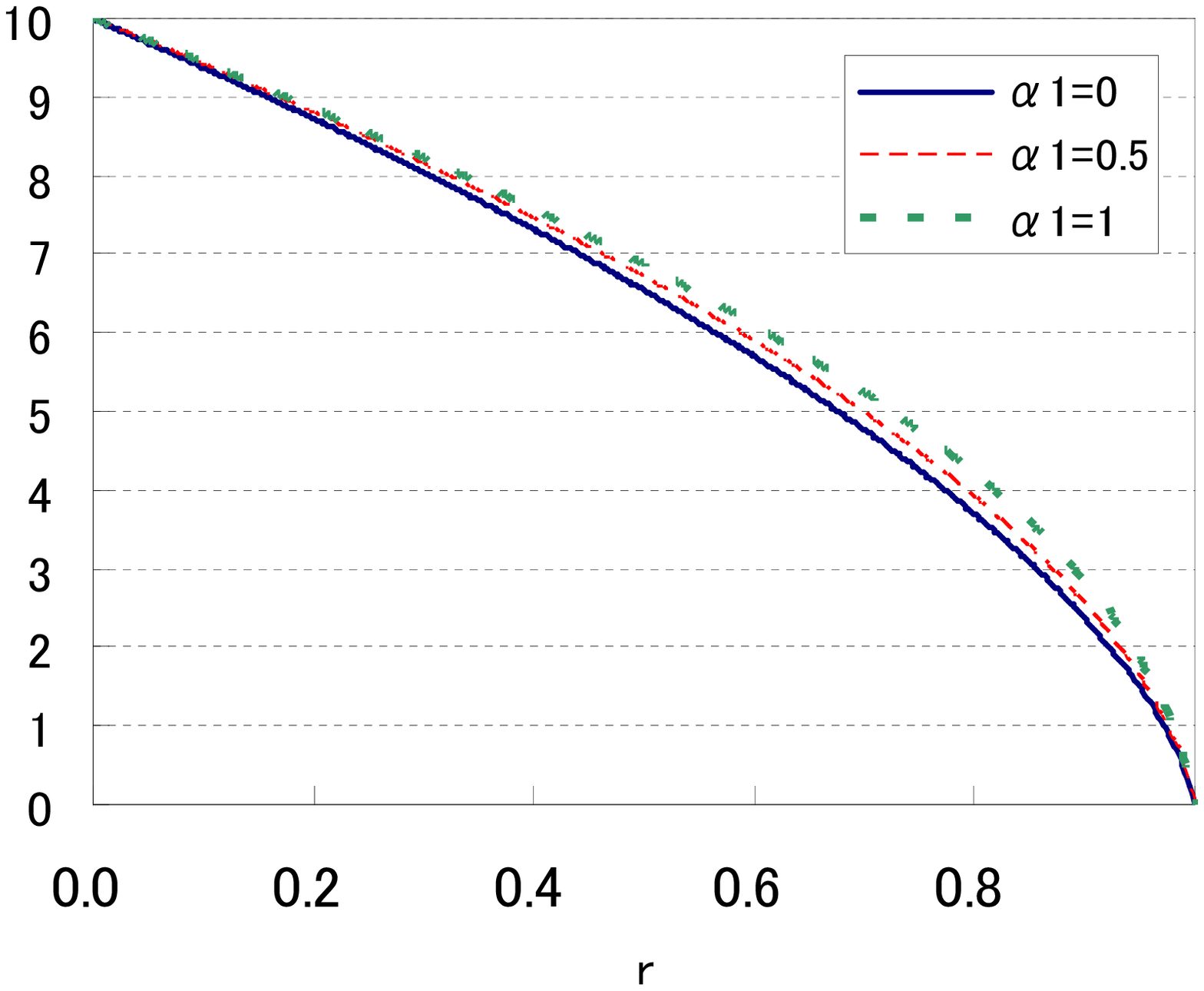}
\caption{Result for $\varphi = 10$ in the case of fixed $\tilde{\gamma }$. 
Left: The optimal strategy $\zeta _r$. 
Right : The amount of security holdings $\varphi _r$. 
}
\label{graph_phi10}
\includegraphics[scale=0.3]{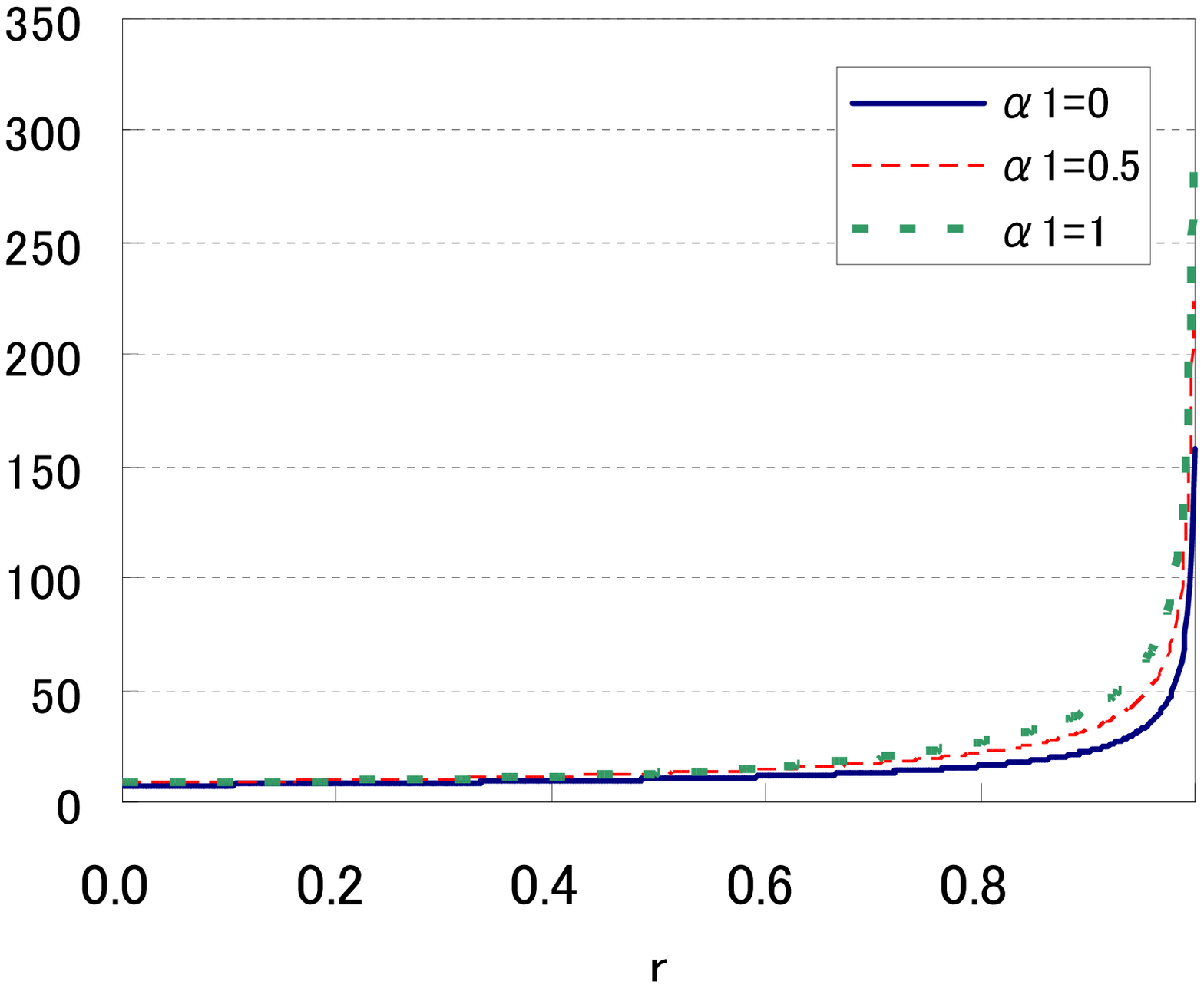}
\includegraphics[scale=0.3]{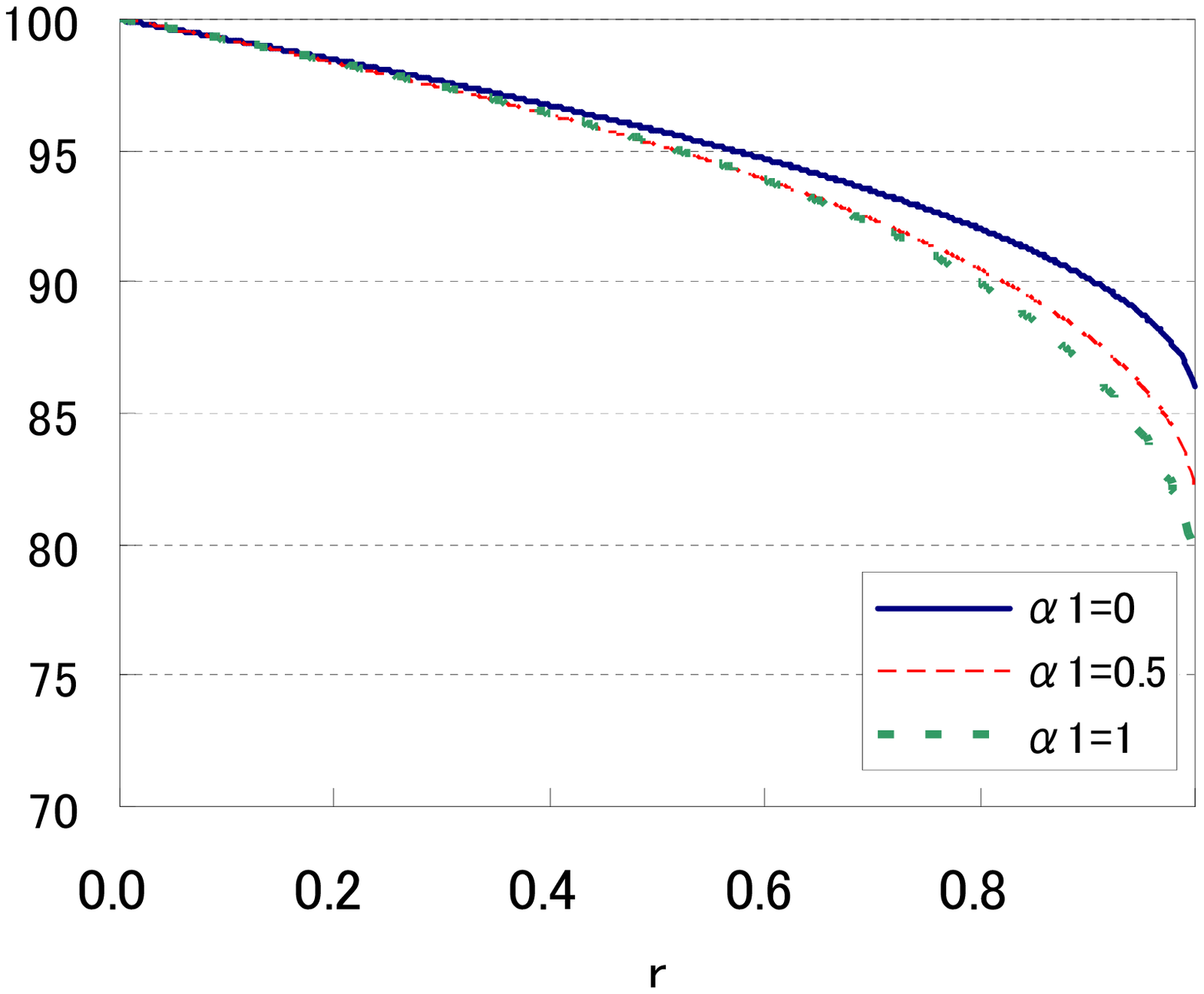}
\caption{Result for $\varphi = 100$ in the case of fixed $\tilde{\gamma }$. 
Left : The optimal strategy $\zeta _r$. 
Right : The amount of security holdings $\varphi _r$. 
}
\label{graph_phi100}
\end{center}
\end{figure}

\begin{figure}[htbp]
\begin{center}
\includegraphics[scale=0.32]{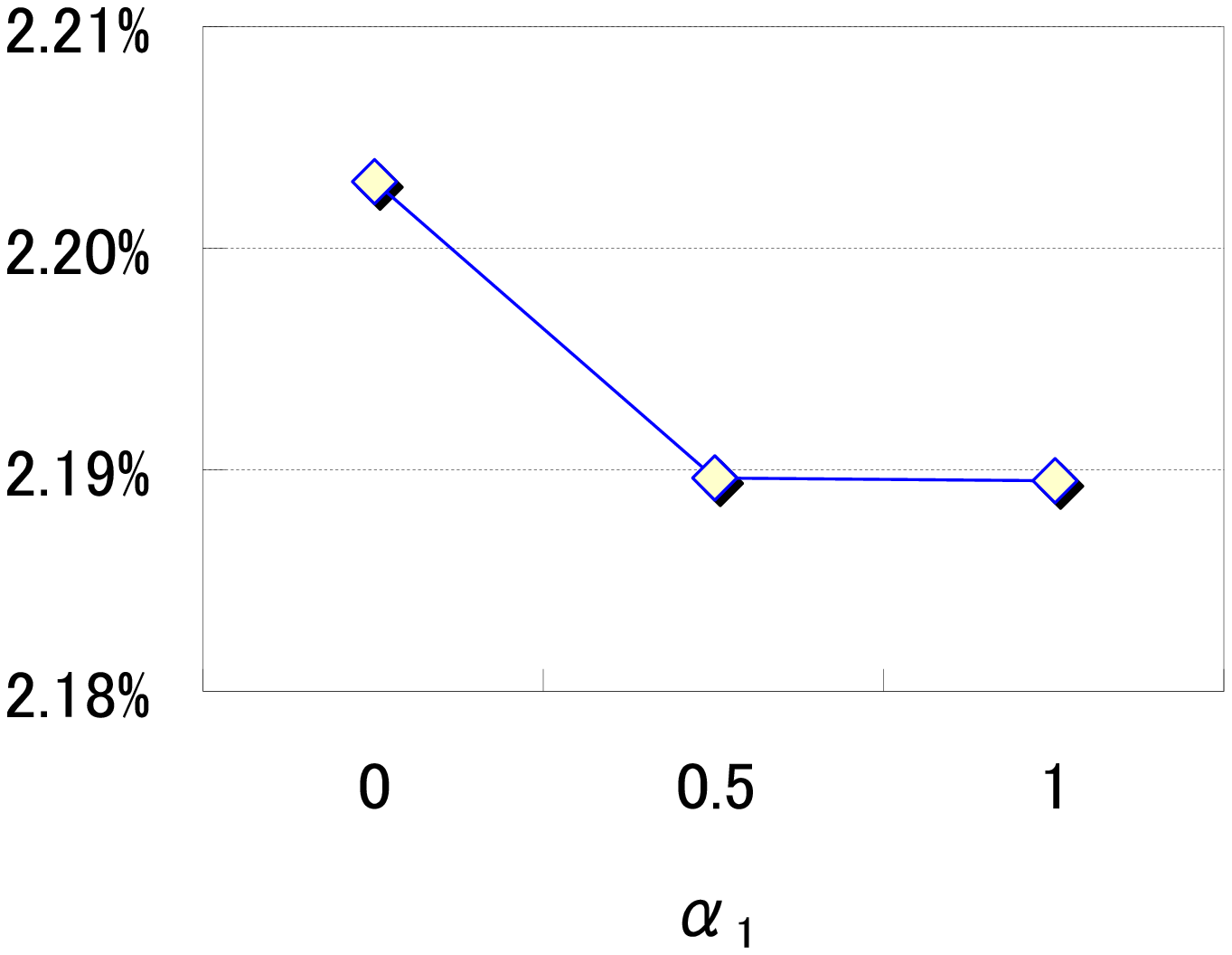}
\includegraphics[scale=0.32]{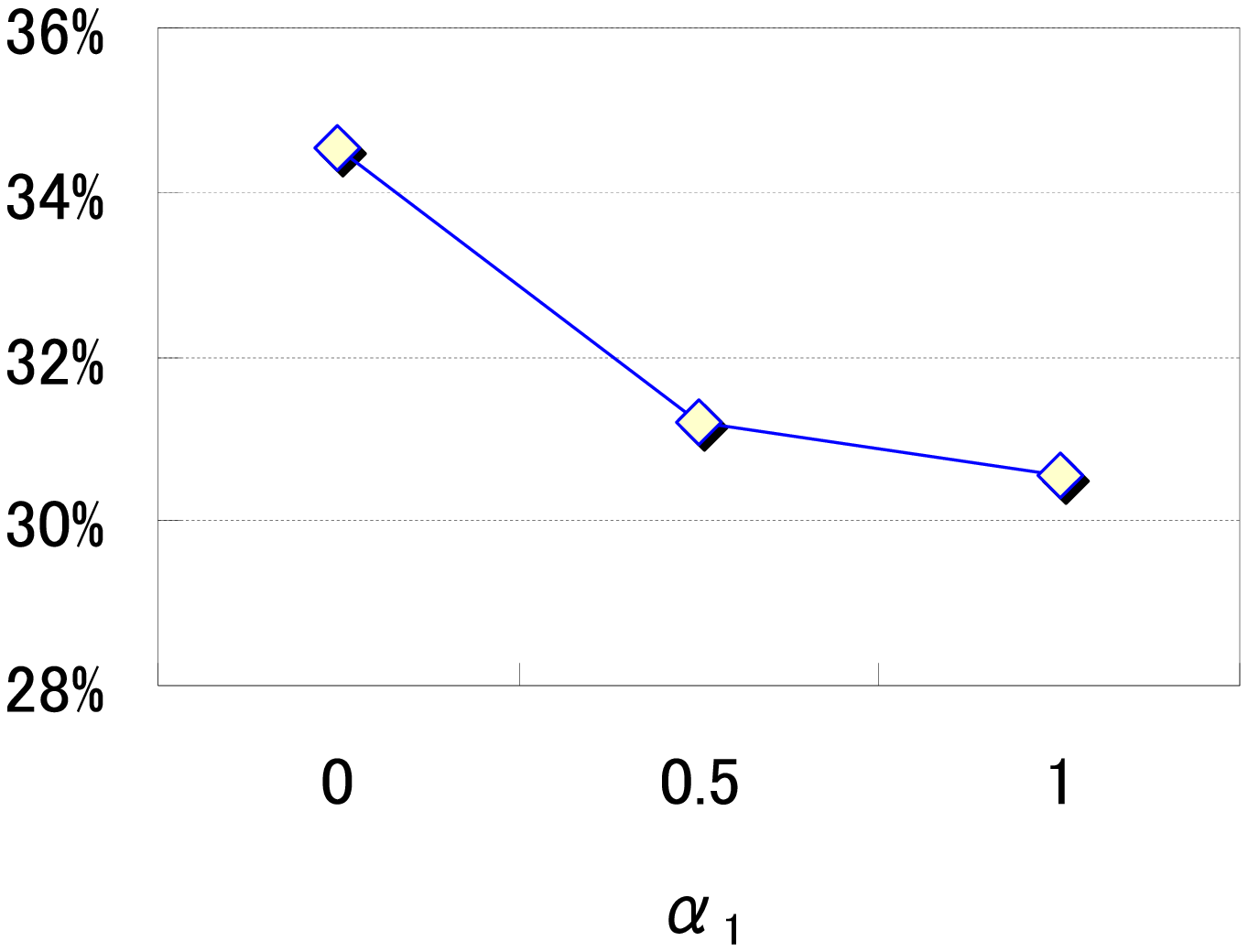}
\caption{Total MI cost $\mathrm {TC}(\varphi)$ for a risk-neutral trader. 
Left: the case of $\varphi = 1$. 
Right: the case of $\varphi = 10$. 
The horizontal axes denote the shape parameter $\alpha _1$ of the Gamma distribution.}
\label{graph_MIcost}
\end{center}
\end{figure}

\section{Concluding Remarks}\label{section_conclusion}

In this paper, 
we studied an optimal execution problem with uncertain MI by using the model derived in \cite {Ishitani-Kato_COSA1}. 
Our main results discussed in Sections \ref {section_Results} and \ref {sec_SO} are almost the same as in \cite {Kato}. 

When considering uncertainty in MI, 
there are two typical barometers of the ``level'' of MI: 
$\gamma $ and $\tilde{\gamma }$. 
By using the parameter $\gamma $, 
we can decompose MI into a deterministic part $\gamma g(\zeta _t)dt$ and 
a pure jump part $g(\zeta _t)\int _{(0, \infty )}zN(dt, dz)$. 
Then, the pure jump part can be regarded as the difference from 
the deterministic MI case studied in \cite {Kato}. 
On the other hand, 
as mentioned in Sections \ref {section_comparison} and \ref {section_examples}, 
the parameter $\tilde{\gamma }$ is important not only in martingale theory but also 
in a risk-neutral world. 
Studying $\tilde{\gamma }$ also provides some hints about actual trading practices. 
Regardless of whether we accommodate uncertainty into MI, 
it may result in an underestimate of MI for a risk-neutral trader. 

Studying the effects of uncertainty in MI in a risk-averse world is also meaningful. 
As mentioned in Section \ref {sec_SO}, when the deterministic part of the MI function is linear, 
the uncertainty in MI does not significantly influence the trader's behavior, even when the trader is risk-averse. 
In future work, we will investigate the case of nonlinear MI. 

Explicitly introducing trading volume processes is another important generalization. 
In some studies of the optimization problem of volume-weighted average price (VWAP) slippage, 
the trading volume processes are introduced as stochastic processes. 
For instance, \cite {Frei-Westray} studies a minimization problem of 
the tracking error of VWAP execution strategies (see \cite {Kato_VWAP} for a definition of VWAP execution strategies). 
In \cite {Frei-Westray}, a cumulative trading volume process is 
defined as a Gamma process. 
Moreover, \cite {Kato_VWAP} treats a generalized Almgren--Chriss model such that 
a temporary MI function depends on instantaneous trading volume processes, 
and shows that an optimal execution strategy of a risk-neutral trader is 
actually the VWAP execution strategy. 
Since a trading volume process is unobservable, 
we can regard it as a source of the uncertainty of MI functions. 
Therefore, studying the case where MI functions are affected by trading volumes is within our focus. 

Finally, in our settings the MI function is stationary in time, but in the real market 
the characteristics of MI change according to the time zone. 
Therefore, it is meaningful to study the case where the MI function is not time-homogeneous. 
This is another topic for future work.

\section{Proofs} \label{sec_proof}

We first recall some lemmas from \cite {Ishitani-Kato_COSA1}. 

\begin{lemma} \label{lemm_conti_u} 
Let $\Gamma _k$ $(k\in \Bbb {N})$ be sets, $u\in {\mathcal {C}}$, and let 
$(W^i(k,\gamma ), \varphi ^i(k,\gamma ), S^i(k,\gamma ))\in D$ 
$(\gamma \in \Gamma _k$, $k\in \Bbb {N}$, $i=1,2)$ be random variables. 
Assume that 
\begin{eqnarray*}
&&\lim _{k\rightarrow \infty }\sup _{\gamma \in \Gamma _k} 
\E [|W^1(k,\gamma )-W^2(k,\gamma )|^{m_1} + |\varphi ^1(k,\gamma )-\varphi ^2(k,\gamma )|^{m_2}\\
&&\hspace{57mm} + |S^1(k,\gamma )-S^2(k,\gamma )|^{m_3}] = 0
\end{eqnarray*}
and 
\begin{eqnarray*}
\sum ^2_{i=1}\sup _{k\in \Bbb {N}}\sup _{\gamma \in \Gamma _k}
\E [|W^i(k,\gamma )|^{m_4}+(S^i(k,\gamma ))^{m_4}] < \infty 
\end{eqnarray*}
for some $m_1, m_2, m_3 > 0$ and $m_4 > m_u$, where $m_u$ is as appeared in $(\ref {growth_C})$.
Then we have
\begin{eqnarray*}
&&\lim _{k\rightarrow \infty }\sup _{\gamma \in \Gamma _k}
\big| \E [u(W^1(k,\gamma ), \varphi ^1(k,\gamma ), S^1(k,\gamma ))] \\
&&\qquad \qquad \quad -  
\E [u(W^2(k,\gamma ), \varphi ^2(k,\gamma ), S^2(k,\gamma ))]\big| = 0 .
\end{eqnarray*}
\end{lemma}

\begin{lemma}\label{Lemma_Moment_Estimate} 
Let $Z(t; r, s) = \exp (Y(t; r, \log s))$ and 
$\hat{Z}(s) = \sup_{0\leq r\leq 1}Z(r; 0, s)$. 
Then, for each $m>0$, there is a constant 
$C_{m, K}>0$ depending only on $K$ and $m$ such that 
$E [\hat{Z}(s)^m]\leq C_{m, K} s^m$, 
where $K>0$ is a constant appearing in (\ref{Bdd_Lipschitz_Constant}).
\end{lemma}

\noindent
\begin{lemma} \label{Ishi_Lem}
Let $(X^{k, i}_r)_{r\in [0,1]}$, $i = 1, 2$, $ k \in \Bbb {N}$, 
be $\mathbb{R}$-valued $(\mathcal {F}_r)_r$-progressive processes satisfying 
\begin{eqnarray*}
X^{k, i}_r = x^{k, i} + \int ^r_0b(X^{k, i}_v)dv + \int ^r_0\sigma (X^{k, i}_v)dB_v + F^{k, i}_r , 
\ \ r\in [0,1],
\end{eqnarray*}
with $x^{k, i} \in \Bbb{R}$  
for $i = 1, 2$ and $k \in \Bbb {N}$, where 
$(F^{k, i}_r)_r$ are $(\mathcal {F}_r)_r$-adapted processes of bounded variation, 
and let $\Pi_k \subset [0,1]$, $k \in \Bbb {N}$, be Borel sets. 
Moreover, assume that 
\begin{description}
 \item[(i)] $x^{k, 1} - x^{k, 2}\longrightarrow 0, \ \ k\rightarrow \infty $, 
 \item[(ii)] 
$\lim_{k\to \infty} \left\{ D^k_1+\int_0^1 D^k_r dr \right\}= 0$,  
where 
\begin{eqnarray*}
D^k_r = \E \left [\sup_{v\in \Pi_k(r) }|F^{k, 1}_v - F^{k, 2}_v|\right ], \ \ 
\Pi_k(r)=([0,r]\cap \Pi_k)\cup \{r\}. 
\end{eqnarray*}
\end{description}
Then it holds that 
\begin{eqnarray*}
\E \left [\sup_{v \in \Pi_k } \left| X^{k, 1}_v - X^{k, 2}_v \right| \right ]  
\longrightarrow 0, \ \ k \rightarrow \infty . 
\end{eqnarray*}
\end{lemma}

\begin{lemma}\label{lem_comparison} \ 
Let $t\in [0,1]$, $\varphi \geq 0$, $x\in \Bbb {R}$, 
$(\zeta _r)_{0\leq r\leq t}, (\zeta '_r)_{0\leq r\leq t}
\in \mathcal {A}_t(\varphi )$ and suppose 
$(X_r)_{0\leq r\leq t}$ $($resp., $(X'_r)_{0\leq r\leq t}$$)$ 
is given by $(\ref {SDE_X_g})$ with $(\zeta _r)_r$ 
$($resp., $(\zeta '_r)_r$$)$ and $X_0 = x \leq X'_0$. 
Suppose $\zeta _r\leq \zeta '_r$ for any $r\in [0,t]$ almost surely. 
Then $X_r\geq X'_r$ for any $r\in [0,t]$ almost surely. 
\end{lemma}

\subsection{Proof of Theorem \ref {conti_random} }\label{subsec_conti}
Continuity in $(w, \varphi , s)$ can be easily proved in the same manner as in the previous study \cite {Kato}, 
so we focus on the continuity in $t$ (uniformly on any compact subset of $D$). 

First of all, we prove the following lemma: 
\begin{lemma}\label{eval_0}
Assume $h(\infty ) = \infty $. Then, for any $t\in [0,1]$, $\varphi \in [0, \Phi_0]$, 
and $(\zeta _r)_{0\leq r\leq t}\in \mathcal {A}_t(\varphi )$,
\begin{eqnarray}\label{temp_lemma_1}
\E \Big[ \int ^r_0\exp \Big( -\int ^v_0g(\zeta _{v'})dL_{v'}\Big) \zeta _v dv \Big]
\leq \phi (r), \ \ r\in [0,t], 
\end{eqnarray}
where $\phi (r),\ r\in (0,1]$ is a continuous function depending only on function 
$h(\zeta )$ and $\Phi _0$, 
such that $\lim _{r\rightarrow 0}\phi (r) = 0$.  
\end{lemma}
\begin{proof} [Proof of Lemma \ref {eval_0}]
We may assume that $\tilde{\gamma } > 0$. 
Let $\pi _r = \int ^r_0g(\zeta _v)dL_v$ and 
$\tau _R = \inf \{ v\in [0,t]\ ; \ \pi _v>R \} \wedge r$ 
for $r\in (0,t]$ and $R>0$. Since $(\pi _v)_v$ is nondecreasing and
$(\exp (-\pi _{v-}) \zeta _v)_v$ is left-continuous, we have that
\begin{align}
\E\Big[ \int ^r_0\exp ( -\pi _v) \zeta _v dv \Big] &\leq
\E\Big[ \int ^r_0\exp ( -\pi _{v-}) \zeta _v dv \Big] \nonumber 
=\frac{1}{\tilde{\gamma}}\E\Big[ \int ^r_0\exp ( -\pi _{v-}) \zeta _v dL_v \Big] \\
&\leq \frac{1}{\tilde{\gamma}}\E\Big[ \int ^{(\tau _R+\varepsilon)\wedge r}_0 \zeta _v dL_v \Big] 
+\frac{e^{-R}}{\tilde{\gamma}} \E\Big[ \int ^r_{(\tau _R+\varepsilon)\wedge r}\zeta _v dL_v \Big] 
\label{Ineq_E_int_0_to_r_exp_mPIv_zeta_v_dv}
\end{align}
holds for $r\in (0,t]$, $R>0$ and $\varepsilon >0$. 
Using the left-continuity of $(\zeta _v)_v$, we obtain
\begin{align*}
\frac{e^{-R}}{\tilde{\gamma}} \E\Big[ \int ^r_{(\tau _R+\varepsilon)\wedge r}\zeta _v dL_v \Big] 
\leq \frac{e^{-R}}{\tilde{\gamma}} \E\Big[ \int ^r_{0}\zeta _v dL_v \Big] 
=e^{-R}\int ^r_{0}\E [ \zeta _v ] dv  \leq \Phi _0e^{-R}. 
\end{align*}
The first term on the right side of (\ref {Ineq_E_int_0_to_r_exp_mPIv_zeta_v_dv}) is rewritten as 
\begin{align}\label{Ineq_E_int_0_to_tauR_zeta_v_dv}
\frac{1}{\tilde{\gamma}}\E\Big[ \int ^{(\tau _R+\varepsilon)\wedge r}_0 \zeta _v dL_v \Big] 
=r \E\Big[ \int^r_0 \zeta _v 1_{[0, \tau _R+\varepsilon]}(v) \frac{dL_v}{\tilde{\gamma}r}\Big]. 
\end{align}
Since $g(\zeta )$ is convex and 
$(\tilde{\gamma} r)^{-1}dL_v(\omega)P(d\omega)$ 
is a probability measure on \\
$([0, r]\times \Omega, \mathcal{B}([0,r])\otimes \mathcal{F})$, 
we apply the Jensen inequality to obtain 
\begin{align*}
& g\Big( \E\Big[ \int^r_0 \zeta _v 1_{[0, \tau _R+\varepsilon]}(v) \frac{dL_v}{\tilde{\gamma}r}\Big] \Big)
\leq \E\Big[ \int^r_0 g(\zeta _v 1_{[0, \tau _R+\varepsilon]}(v)) \frac{dL_v}{\tilde{\gamma}r}\Big]
=\frac{\E [ \pi_{(\tau _R+\varepsilon)\wedge r}]}{\tilde{\gamma}r} .
\end{align*}
Combining this with (\ref{Ineq_E_int_0_to_tauR_zeta_v_dv}) 
we get 
\begin{align*}
&\frac{1}{\tilde{\gamma}}\E\Big[ \int ^{(\tau _R+\varepsilon)\wedge r}_0 \zeta _v dL_v \Big] 
\leq r g^{-1}\left( \frac{\E [ \pi_{(\tau _R+\varepsilon)\wedge r}]}{\tilde{\gamma}r} \right) ,
\end{align*}
where $g^{-1}(y) := \sup \{ \zeta \in [0,\infty ) \ ; \ g(\zeta ) = y \} $, $y\geq 0$. 
Since $\big( \int^v_0 \zeta_{v'}dL_{v'} \big)_v$ and $(\pi_v)_v$ are right-continuous, 
and $g^{-1}(y)$ is a continuous function on $y\in [0, \infty)$, 
we have that
\begin{align*}
&\frac{1}{\tilde{\gamma}}\E\Big[ \int ^{\tau _R}_0 \zeta _v dL_v \Big] 
\leq 
\lim_{\varepsilon \to 0}
r g^{-1}\left( \frac{\E [ \pi_{(\tau _R+\varepsilon)\wedge r}]}{\tilde{\gamma}r} \right)
=r g^{-1}\left( \frac{\E [ \pi_{\tau _R}]}{\tilde{\gamma}r} \right) 
\leq r g^{-1}\left( \frac{R}{\tilde{\gamma}r} \right) . 
\end{align*}
Summarizing the above arguments, we arrive at 
\begin{eqnarray*}
\E \left[ \int ^r_0\exp ( -\pi _v) \zeta _v dv\right]  \leq 
r g^{-1}\left( \frac{R}{\tilde{\gamma}r}\right) + \Phi _0e^{-R}. 
\end{eqnarray*}
Therefore, if we can find a positive function $R(r)$ that satisfies 
\begin{eqnarray}\label{cond_Rr}
R(r) \longrightarrow \infty \ \ \mathrm {and} \ \ 
r g^{-1}\left (\frac{R(r)}{\tilde{\gamma}r}\right ) \longrightarrow 0\ \ \mathrm {as} \ \ r\rightarrow 0, 
\end{eqnarray}
we complete the proof of (\ref {temp_lemma_1}). 
To construct such an $R(r)$, mimicking the proof of Lemma B.12 in \cite {Kato}, we define 
\begin{eqnarray*}
R(r) = \tilde{\gamma} r g(M(r)), \ \ 
M(r) = f^{-1}\left( \frac{1}{r}\right), \ \ 
f(\zeta ) = \zeta \sqrt{h\left( \frac{\zeta }{2}\right) }, \ \ r > 0, 
\end{eqnarray*}
where the inverse function $f^{-1}(y)$ is defined in the same manner as $g^{-1}(y)$. 
We can easily verify (\ref {cond_Rr}) by the same arguments as in \cite {Kato}. 
\end{proof}

The following proposition can be proved by the same proof as Theorem 3.1(ii) in \cite {Kato} 
in combination with Lemma \ref{eval_0} and Proposition \ref {th_semi}. 
\begin{proposition}
Assume $h(\infty)=\infty$. Then for any compact set $E\subset D$, 
\begin{align*}
\lim _{t\downarrow 0} \sup _{(w,\varphi ,s)\in E} 
\vert  V_{t}(w,\varphi ,s;u)-u(w, \varphi ,s)\vert = 0.
\end{align*}
\end{proposition}

Next we consider the case where $h(\infty)<\infty$. 
Hereinafter, for each $(w, \varphi, s) \in D$ and $(\zeta _r)_r\in \mathcal {A}_t(\varphi )$, 
we denote by $\Xi _t(w,\varphi ,s ; (\zeta _r)_r)$ the ordered triplet of processes 
$(W_r, \varphi _r, S_r)_{0\leq r\leq t}$ given by the differential equations in (\ref {SDE_X_g}).

\begin{proposition}\label{thm2_pro1} 
Assume $h(\infty ) < \infty $. Then for any compact set  $E\subset D$ we have 
\begin{eqnarray*}
\limsup _{t\downarrow 0 } \sup _{(w,\varphi ,s)\in E}
(Ju(w,\varphi ,s)-V_t(w,\varphi ,s;u))\leq 0. 
\end{eqnarray*}
\end{proposition}
\begin{proof}
Take any $t\in (0,1)$, $(w,\varphi, s)\in E$, and $\psi \in [0, \varphi]$. 
Set $(\zeta _r)_r \in \mathcal{A}_t(\varphi )$ by 
$\zeta _r = \frac{\psi }{t}\ (0\leq r\leq t)$, 
and let $(W_r,\varphi _r, S_r)_{0\leq r\leq t}=\Xi _t(w, \varphi, s; (\zeta _r)_r)$ and $X_r=\log S_r$. 
A standard argument leads us to 
\begin{align*}
\E\Big[ \sup_{r\in [0, t]}\vert \exp (X_r) - s\exp \left(-g({\psi}/{t})L_r \right) \vert \Big]
&\leq C_K s \sqrt{t}, \\
\E\left[ \big\vert W_t - w - \psi s \int_0^1 \exp \left(-g({\psi}/{t})L_{tv} \right) dv \big\vert \right] 
&\leq C_K \psi s \sqrt{t}
\end{align*}
for some $C_K > 0$. 
Thus, using Lemma \ref{lemm_conti_u}, we get 
\begin{align}\nonumber 
&\sup_{\substack{(w, \varphi ,s)\in E \\ \psi \in [0,\varphi ]}}
\left\{ 
I_1((\zeta _r)_r) -
V_t(w,\varphi , s;u) \right\} \\
&\quad \leq \sup_{\substack{(w, \varphi ,s)\in E \\ \psi \in [0,\varphi ]}}
\left\{
I_1((\zeta _r)_r)
- \E [u(W_t,\varphi_t , \exp(X_t)) ]
\right\} \longrightarrow 0 \ \ \,\,t \downarrow 0, 
\label{prop6_conv1} 
\end{align}
where 
\begin{eqnarray*}
I_1((\zeta _r)_r) = \E [u(w+\psi s\int_0^1 \exp (-g(\psi /t)L_{tv})dv,\varphi - \psi ,
s\exp \left(-g(\psi /t)L_t \right)) ]. 
\end{eqnarray*}

Next we will show 
\begin{eqnarray}\label{LOLN}
\sup_{\substack{(w, \varphi ,s)\in E \\ \psi \in [0,\varphi ]}}
\left| I_1((\zeta _r)_r) - I_2((\zeta _r)_r)\right | 
\longrightarrow 0, \ \ t\downarrow 0, 
\end{eqnarray}
where 
\begin{align*}
I_2((\zeta _r)_r) = \E \Big[ u\Big( w+\psi s\int_0^1 \exp (-g(\psi /t)\gamma tv)dv,\varphi - \psi ,
s\exp \left(-g(\psi /t)\gamma t \right) \Big) \Big]. 
\end{align*}
Theorem 9.43.20 in \cite {Sato} implies 
\begin{eqnarray}\label{conv_L}
\lim _{t\downarrow 0 }\frac{L_t}{t} = \gamma  \ \ \mathrm {a.s.} 
\end{eqnarray}
Hence, we obtain 
\begin{align*}
&\sup_{\substack{(w, \varphi ,s)\in E \\ \psi \in [0,\varphi ]}}
\E \left[ \vert \exp(-g(\psi /t)\gamma t) - \exp(-g(\psi /t)L_t) \vert \right] \\
&\quad \leq \E\left[ 1 - \exp\left( t g(\varphi ^* /t) 
\left\{ \gamma - \frac{L_t}{ t} \right\}\right) \right]
\longrightarrow 0, \ \ t\downarrow 0, 
\end{align*}
where we denote $\varphi ^* := \sup _{(w, \varphi , s)\in E} \varphi$. 
Similarly, we obtain 
\begin{equation*}
\lim_{t\downarrow 0 }\sup_{\substack{(w, \varphi ,s)\in E \\ \psi \in [0,\varphi ]}}
\E \left[ \left\vert \psi s \int_0^1 
\left\{ \exp(-g(\psi /t)\gamma tv) -\exp(-g(\psi /t) L_{tv}) \right\} dv
\right\vert \right] = 0 .
\end{equation*}
Thus we get (\ref{LOLN}) by using Lemma $\ref {lemm_conti_u}$. 

We now complete the proof of Proposition \ref {thm2_pro1}. 
By the monotonicity of $u(w,\varphi ,s)$ (especially in $w$ and $s$) and the inequality 
$(0\leq ) t g(\psi /t)\leq \psi h(\infty)$, 
we see that 
\begin{align*}
I_2((\zeta _r)_r) \geq u(w+F(\psi)s,\varphi - \psi , s e^{-\gamma h(\infty )\psi }) , 
\end{align*}
where 
\begin{eqnarray*}
F(\psi) = \int_0^{\psi}e^{-\gamma h(\infty )p}dp 
= \psi \int^1_0 \exp (-\gamma h(\infty) \psi v) dv. 
\end{eqnarray*}
Therefore, 
\begin{align}\label{temp_F0}
\sup _{(w,\varphi ,s)\in E}(Ju(w,\varphi ,s)-V_t(w,\varphi ,s;u)) 
\leq 
\sup_{\substack{(w, \varphi ,s)\in E \\ \psi \in [0,\varphi ]}}\left( I_2((\zeta _r)_r) - \E [u(W_t, \varphi _t, S_t)] \right) . 
\end{align}
Now our assertion is shown immediately from (\ref {prop6_conv1}), (\ref{LOLN}), and (\ref {temp_F0}). 
\end{proof}

\vspace{0,5cm}
\begin{proposition}\label{thm2_pro2} 
Assume $h(\infty ) < \infty $. Then for any compact set $E\subset D$ , 
\begin{eqnarray*}
\limsup _{t\downarrow 0 } \sup _{(w,\varphi ,s)\in E}
(V_t(w,\varphi ,s;u)-Ju(w,\varphi ,s))\leq 0. 
\end{eqnarray*}
\end{proposition}
\begin{proof}
Take any $t \in (0, 1)$, $(w, \varphi , s) \in E$, and $(\zeta_r)_r \in \mathcal{A}_t(\varphi)$. 
Denote\\ 
$(W_r,\varphi _r, \allowbreak S_r)_{0\leq r\leq t}=\Xi _t(w, \varphi , s ; (\zeta _r)_r)$ 
and $X_r=\log S_r$. 
Since $g$ is convex, the Jensen inequality implies
\begin{align*}
\int_0^r g(\zeta _v)dL_v \geq 
\gamma \int_0^r g(\zeta _v)dv
\geq \gamma r g\left( \frac{1}{r} \int_0^r \zeta_v dv\right)
= \gamma \int_0^{\eta_r}h(\zeta / r)d\zeta ,\quad r\in [0, t],
\end{align*}
where $\eta _r = \int_0^r \zeta_v dv $. 
Then we have 
\begin{align}\label{Th2_Uniform_Estimate1}
&u\Big( w + s \int_0^t \zeta_r \exp \big(-\int_0^r g(\zeta _v)dL_v \big) dr, 
\varphi - \eta_t, s e^{-\int_0^t g(\zeta _v)dL_v } \Big) \nonumber \\ 
&\quad \leq 
u\Big( w + s \int_0^t \zeta_r \exp \big( -\gamma \int_0^{\eta_r}h(\zeta/r)d\zeta \big)dr , 
\varphi - \eta_t, s e^ {-\gamma \int_0^{\eta_t}h(\zeta /t)d\zeta } \Big) . 
\end{align}

As in the proof of Proposition \ref {thm2_pro1}, we get 
\begin{align}
\E\Big[ \sup_{r\in [0, t]}\Big\vert \exp (X_r) - s\exp \big(-\int_0^r g(\zeta _v)dL_v \big) \Big\vert \Big]
&\leq C_K s \sqrt{t}, \label{Th2_Uniform_Estimate2}\\
\E\left[ \Big\vert W_t - w - s \int_0^t \zeta_r \exp \big(-\int_0^r g(\zeta _v)dL_v \big) dr \Big\vert \right] 
&\leq C_K \Phi_0 s \sqrt{t} \label{Th2_Uniform_Estimate3}
\end{align}
for some $C_K > 0$. 
Then we can apply Lemma \ref{lemm_conti_u} with 
(\ref{Th2_Uniform_Estimate2}) and (\ref{Th2_Uniform_Estimate3}) to obtain
\begin{align}
&\sup_{\substack{(w, \varphi ,s)\in E \\ (\zeta_r)_r \in \mathcal{A}_t(\varphi)}}
\Big\vert
\E\Big[ u\Big( w + s \int_0^t \zeta_r \exp \big(-\int_0^r g(\zeta _v)dL_v \big) dr, 
\varphi - \eta_t, s e^{-\int_0^t g(\zeta _v)dL_v } \Big)\Big]
\nonumber \\
&\qquad \qquad \quad -
\E[u(W_t, \varphi_t, S_t)]
\Big\vert 
\longrightarrow 0\ \ \mathrm {as} \ \ t \downarrow 0. \label{Th2_Uniform_Conv1}
\end{align}

We can also see that 
\begin{align}
\sup_{r\in [0, t]}\Big\vert \exp \left( - \gamma \int_0^{\eta_r} h(\zeta / r) d\zeta \right)
- e^{- \gamma h(\infty ) \eta_r} \Big\vert 
&\leq 2\gamma \widetilde{\varepsilon}_t , \label{Th2_Uniform_Estimate4}\\
\Big\vert \E \Big[ \int_0^t \zeta_r \Big\{
\exp \Big( -\gamma \int_0^{\eta_r}h(\zeta / r)d\zeta \Big)
-e^{- \gamma h(\infty ) \eta_r} \Big\} dr \Big] \Big\vert 
&\leq 2\gamma \Phi_0 \widetilde{\varepsilon}_t , \label{Th2_Uniform_Estimate5}
\end{align}
where 
$\widetilde{\varepsilon}_t = \int_0^{\Phi_0} \big( h(\infty) - h(\zeta /t) \big)d\zeta 
(\longrightarrow 0, \ \ t\downarrow 0)$. 
Applying Lemma \ref{lemm_conti_u} again with (\ref{Th2_Uniform_Estimate4}) and (\ref{Th2_Uniform_Estimate5}), 
we have that 
\begin{align}
&\sup_{\substack{(w, \varphi ,s)\in E \\ (\zeta_r)_r \in \mathcal{A}_t(\varphi)}}
\Big\vert
\E\Big[ u\Big( w + s \int_0^t \zeta_r \exp \Big( -\gamma \int_0^{\eta_r}h(\zeta / r)d\zeta \Big)dr , 
\varphi - \eta_t, s e^{-\gamma \int_0^{\eta_t}h(\zeta / t)d\zeta } \Big) \Big]\nonumber \\
&\qquad \qquad \quad - \E\Big[
u\Big(w + s \int_0^t \zeta_r e^{- \gamma h(\infty ) \eta_r} dr, 
\varphi - \eta_t, s e^{- \gamma h(\infty ) \eta_t} \Big) \Big]
\Big\vert \longrightarrow 0\ \ \mathrm {as} \ \ t \downarrow 0. \label{Th2_Uniform_Conv2}
\end{align}
Moreover, from the definition of $Ju(w, \varphi , s)$, we see that 
\begin{align}\nonumber 
&\sup_{\substack{(w, \varphi ,s)\in E \\ (\zeta_r)_r \in \mathcal{A}_t(\varphi)}}
\Big\{
\E\Big[ u\Big(w + s \int_0^t \zeta_r e^{- \gamma h(\infty ) \eta_r} dr, 
\varphi - \eta_t, s e^{- \gamma h(\infty ) \eta_t}\Big) \Big]
-Ju(w, \varphi , s) \Big\} \\
&\quad = \sup_{\substack{(w, \varphi ,s)\in E \\ (\zeta_r)_r \in \mathcal{A}_t(\varphi)}}
\Big\{
\E\Big[ u\Big(w + s F(\eta _t), \varphi - \eta_t, s e^{- \gamma h(\infty ) \eta_t}\Big) \Big] - Ju(w, \varphi , s) 
\Big\} \leq 0.
\label{Th2_Uniform_Ineqality3}\end{align}
Combining (\ref{Th2_Uniform_Estimate1}), (\ref{Th2_Uniform_Conv1}), 
(\ref{Th2_Uniform_Conv2}), and 
(\ref{Th2_Uniform_Ineqality3}), we obtain our assertion.
\end{proof}

Finally, we consider the continuity with respect to $t\in (0, 1]$. 
\begin{proposition}\label{thm2_pro3}
Let $E\subset D$ be a compact set.  Then we have the following: \\
$\mathrm {(i)}$ \ \ $\lim _{t'\uparrow t }
\sup _{(w,\varphi ,s)\in E}
\vert V_{t'}(w,\varphi ,s;u) - V_{t}(w,\varphi ,s;u)\vert = 0$, \ $t\in (0,1]$.\\
$\mathrm {(ii)}$ \ $\lim _{t'\downarrow t }
\sup _{(w,\varphi ,s)\in E}
\vert V_{t'}(w,\varphi ,s;u) - V_{t}(w,\varphi ,s;u)\vert = 0$, \ $t\in (0,1)$.
\end{proposition}
\begin{proof}
All we have to do is to show that 
\begin{eqnarray}
JV_t(w,\varphi,s;u)\leq V_t(w,\varphi,s;u) , 
\quad (w, \varphi , s) \in D , \quad t \in (0, 1) 
\label{obj_t_ineq}
\end{eqnarray}
under $h(\infty ) < \infty $, 
because all the other assertions are obtained 
in the same way as in the proof of Proposition B.17 in \cite {Kato} 
combined with Proposition \ref {th_semi} and (\ref {obj_t_ineq}). 

Take any $t \in (0, 1)$, $(w, \varphi , s) \in D$, $\psi \in [0,\varphi ]$, 
and $(\zeta_r)_{0\leq r\leq t} \in \mathcal{A}_t(\varphi - \psi )$. 
Define 
$(W_r,\varphi _r, S_r)_{0\leq r\leq t}=\Xi _t(w+F(\psi)s, \varphi - \psi, s e^{-\gamma h(\infty) \psi} ; (\zeta _r)_r)$ 
and $X_r=\log S_r$. 
For any $\delta \in (0,t)$, we define $(\tilde{\zeta} _r)_{0\leq r\leq t}\in \mathcal{A}_t(\varphi)$ by 
$\tilde{ \zeta }_r = (\psi / \delta )1_{[0,\gamma \delta ]} (L_{r-}) + \zeta_r$. 
Note that the admissibility of $(\tilde {\zeta _r})_r$ comes from $L_r \geq \gamma r$. 
Furthermore, we denote $(\tilde{W}_r,\tilde{\varphi }_r,\tilde{S}_r)_{0\leq r\leq t} $ 
$=$ $\Xi _t(w, \varphi , s ; (\tilde{\zeta} _r)_r)$ and $\tilde{X}_r = \log \tilde{S}_r$. 

From the definition, we have that
\begin{align*}
&X_r = \log s + 
\int _0^r \sigma (X_v)dB_v + \int _0^r b(X_v)dv + F^{(\delta), 1}_r, \\
&\tilde{X}_r = \log s + 
\int _0^r \sigma (\tilde{X}_v)dB_v + 
\int _0^r b(\tilde{X}_v)dv + F^{(\delta), 2}_r, \quad \mbox{for}\ r\in [0 ,t], 
\end{align*}
where 
\begin{eqnarray*}
F^{(\delta), 1}_r = -\gamma h(\infty ) \psi -  \int _0^r g(\zeta_v)dL_v, 
\quad F^{(\delta), 2}_r =- \int _0^r g(\tilde{ \zeta }_v)dL_v . 
\end{eqnarray*}
We will apply Lemma \ref{Ishi_Lem} with $F^{(\delta), 1}_r$, $F^{(\delta), 2}_r$, 
and $\Pi^{(\delta)} = [\delta , t]$ to show 
\begin{align}\label{Prop8_IshiLem_result}
\E \Big[ \sup_{r\in [\delta ,t]} |\tilde{X}_r - X_r| \Big] \longrightarrow 0 , \ \ 
\delta \downarrow 0.
\end{align}
Set 
$D^{(\delta)}_r = \E \Big[ \sup_{v\in \Pi^{(\delta)}(r)}
\vert F^{(\delta), 1}_v - F^{(\delta), 2}_v \vert \Big]$. 
Obviously it holds that $\Pi^{(\delta)}(r)=[\delta , r]$ ($r\geq \delta$), 
$\{r \}$ ($r< \delta$) and 
\begin{align*}
D^{(\delta)}_t + \int _0^t D^{(\delta)}_r dr 
\leq (2-\delta) \E \Big[ \sup_{v\in [\delta, t]}\vert F^{(\delta), 1}_v - F^{(\delta), 2}_v \vert \Big] 
+ \int _0^{\delta} \E[\vert F^{(\delta), 1}_r - F^{(\delta), 2}_r \vert] dr . 
\end{align*}
Since $(L_v)_v$ is nondecreasing, we see that 
\begin{align*}
\tilde{u}(\delta):=\sup \{v\in [0, t]; L_{v-}\leq \gamma \delta \} 
= \sup \{v\in [0, t]; L_{v}\leq \gamma \delta \} . 
\end{align*}
Moreover, $\tilde{u}(\delta)\leq \delta $ holds from the definition of $(L_r)_r$. 
Then we have 
\begin{align}
&F^{(\delta), 2}_r-F^{(\delta), 1}_r = \gamma h(\infty ) \psi 
- \frac{1}{\delta} \int^{r\wedge \tilde{u}(\delta)}_0 
\Big\{ \int^{\psi}_0 h\left (\frac{1}{\delta}\zeta ' + \zeta_v\right ) d\zeta ' \Big\} dL_v
\label{Diff_F_delta_eq_1}\\
&\quad =h(\infty ) \psi \Big\{ \gamma - \frac{L_{r\wedge \tilde{u}(\delta)}}{\delta}\Big\} 
+ \frac{1}{\delta} \int^{r\wedge \tilde{u}(\delta)}_0 \Big\{ \int^{\psi}_0 
\Big( h(\infty) - h\left (\frac{1}{\delta}\zeta ' + \zeta_v\right ) \Big) d\zeta ' \Big\} dL_v
\nonumber 
\end{align}
for $0\leq r \leq t$. From (\ref{Diff_F_delta_eq_1}), we have 
\begin{align}
&\E \Big[ \sup_{v\in [\delta, t]}\vert F^{(\delta), 1}_v - F^{(\delta), 2}_v \vert \Big] \nonumber \\
&\quad \leq h(\infty)\psi\E \left [ \gamma - \frac{ L_{\tilde{u}(\delta)}}{\delta }\right ] + 
\gamma \int^{\psi}_0 \left( h(\infty)-h\left (\frac{1}{\delta}\zeta '\right )\right ) d\zeta ' ,
\label{temp_last_ineq_1}\\
&\int _0^{\delta} \E[\vert F^{(\delta), 1}_r - F^{(\delta), 2}_r \vert] dr
\leq \delta h(\infty ) \psi \gamma 
+ \delta \tilde{\gamma} \int^{\psi}_0 \left ( h(\infty)-h\left (\frac{1}{\delta}\zeta '\right )\right ) d\zeta '. 
\label{temp_last_ineq_2}
\end{align}
The second terms of the right sides of both (\ref {temp_last_ineq_1}) and 
(\ref {temp_last_ineq_2}) converge to $0$ as $\delta \downarrow 0$. 
Moreover we can show the following lemma: 
\begin{lemma}\label{temp_tilde_u}
$\frac{\tilde{u}(\delta )}{\delta } \longrightarrow 1, \ \ \delta \downarrow 0$ a.s. 
\end{lemma}
By the above lemma and (\ref {conv_L}), 
we have 
\begin{align}\label{conv_Lu}
\frac{L_{\tilde{u}(\delta )}}{\delta}
\longrightarrow \gamma, \ \ \delta \downarrow 0 
\ \ \mathrm {a.s.} 
\end{align} 
Then the dominated convergence theorem implies that the first term of the right side of 
(\ref {temp_last_ineq_1}) also converges to $0$ as $\delta \downarrow 0$. 
Now we arrive at 
\begin{eqnarray*}
D^{(\delta)}_t + 
\int _0^t D^{(\delta)}_r dr \longrightarrow 0 , \ \ 
\delta \downarrow 0, 
\end{eqnarray*}
which immediately implies (\ref {Prop8_IshiLem_result}) together with Lemma \ref {Ishi_Lem}. 

A standard argument with (\ref {Prop8_IshiLem_result}) gives 
\begin{align}
&\E \Big[\sup _{r\in [\delta ,t]} |\exp(\tilde{X}_r)-\exp(X_r)|^{1/2}\Big]
\nonumber \\
&\quad \leq 
(2sC_{1, K})^{1/2}
\E \Big[\sup _{r\in [\delta ,t]} |\tilde{X}_r - X_r|\Big]^{1/2}
\longrightarrow 0, \ \ \delta \downarrow 0 . \label{Prop8_Conv_E_Diff_exp_X}
\end{align}
On the other hand, we see that 
\begin{align*}
\E[ \vert W_t - \tilde{W}_t \vert ^{1/2} ]  
& \leq 
J_1 + J_2 + J_3, 
\end{align*}
where 
\begin{align*}
J_1 &= \E\Big[\Big\vert 
\frac{\psi}{\delta} \int^{\tilde{u}(\delta )}_0 \exp(\tilde{X}_r)dr
-s\int^{\psi}_0 e^{-\gamma h(\infty) p} dp
\Big\vert ^{1/2} \Big], \\
J_2 &= E\Big[ \Big\{ \int^t_{\delta} \zeta_r \vert \exp (\tilde{X}_r)- \exp (X_r) \vert dr \Big\}^{1/2} \Big] , \\
J_3 &= \E\Big[ \Big\{ \int^{\delta}_0 \zeta_r \vert \exp (\tilde{X}_r)- \exp (X_r) \vert dr \Big\}^{1/2} \Big] .
\end{align*}
Easily we get 
\begin{align}
&J_2 
\leq \sqrt{\varphi - \psi} \E \Big[ \sup _{r\in [\delta ,t]} \big|e^{\tilde{X}_r}-e^{X_r}\big|^{1/2}\Big] 
\longrightarrow 0, \ \ \delta \downarrow 0, 
\nonumber \\
&J_3 
\leq (\delta \|\zeta \|_{\infty})^{1/2}
\E \Big[ \sup _{r\in [0, \delta]} \{ e^{\tilde{X}_r}+e^{X_r} \}^{1/2}\Big]
\longrightarrow 0, \ \ \delta \downarrow 0 
\nonumber
\end{align}
by virtue of (\ref {Prop8_Conv_E_Diff_exp_X}) and Lemma \ref {Lemma_Moment_Estimate}. 
As for $J_1$, a similar calculation to (\ref{Diff_F_delta_eq_1}) gives 
\begin{align}\nonumber 
J_1 
&\quad \leq \sqrt{sC_{1, K}\psi }\E \Big[1-\frac{\tilde{u}(\delta )}{\delta }\Big]^{1/2}\nonumber \\
&\qquad + \sqrt{\psi } \E\left[ \left( \frac{1}{\delta}\int^{\delta}_0 
\left\vert \exp (\tilde{X}_r) - s \exp \left(-\frac{\gamma h(\infty) \psi r}{\delta}\right )\right\vert dr  
\right)^{1/2} \right] \nonumber \\
&\quad \leq 
\sqrt{sC_{1, K}\psi } \E \Big[1-\frac{\tilde{u}(\delta )}{\delta }\Big]^{1/2} + 
\sqrt{ s(1+C_{1, K})\psi } \left\{ A^{1/2}_1 + A^{1/2}_2\right\} , 
\label{Prop8_decomp2_Diff_W} 
\end{align}
where 
\begin{align*}
A_1 &= \frac{1}{\delta}\E \left [ \int^{\delta}_0 \left \{ 
\Big\vert \int^r_0 \sigma (\tilde{X}_v)dB_v \Big\vert 
+\Big\vert \int^r_0 b(\tilde{X}_v)dv \Big\vert 
+\Big\vert \int^r_0 g(\zeta_v)dL_v \Big\vert \right \} dr\right] , \\
A_2 &= \frac{1}{\delta}\E \left [  \int^{\delta}_0 
\Big\vert \int^r_0 (g(\tilde{ \zeta }_v)-g(\zeta_v)) dL_v - \frac{\gamma h(\infty)\psi r}{\delta}\Big\vert
dr \right ] . 
\end{align*}
Straightforward calculations lead us to 
\begin{eqnarray}\label{Prop8_decomp4_Diff_W} 
A_1 \leq 
\frac{2K}{3}\sqrt{\delta } + \frac{(K + \tilde{\gamma }g(||\zeta ||_\infty ))\delta }{2} .
\end{eqnarray}
Moreover, 
by Lemma \ref {temp_tilde_u} and (\ref {conv_Lu}), we see that 
\begin{align}
A_2
&\quad \leq \gamma \int^{\psi}_0 (h(\infty)-h(\zeta'/\delta))d\zeta'
+\frac{ \psi h(\infty)}{\delta} \E \Big[ \frac{1}{\delta } \int^{\delta }_0 
\vert \gamma r -L_{r\wedge \tilde{u}(\delta )}\vert dr \Big] 
\nonumber \\
&\quad \leq \gamma \int^{\psi}_0 (h(\infty)-h(\zeta'/\delta))d\zeta'
+\psi h(\infty) \E \Big[ \frac{1}{\delta} \int^{\tilde{u}(\delta )}_0 \Big\{ \frac{L_r}{r}-\gamma \Big\} dr \Big]
\nonumber \\
&\qquad + \psi h(\infty) \E \Big[ \Big( 1-\frac{\tilde{u}(\delta )}{\delta }\Big)
\Big\{ \gamma \Big( 1-\frac{\tilde{u}(\delta )}{\delta }\Big) 
+ \Big(\gamma - \frac{L_{\tilde{u}(\delta )}}{\delta}\Big)\Big\}\Big]
\longrightarrow 0, \ \ \delta \downarrow 0 . \label{Prop8_decomp6_Diff_W} 
\end{align}
Combining Lemma \ref{Lemma_Moment_Estimate}, Lemma \ref{lem_comparison}, 
(\ref{Prop8_decomp2_Diff_W}), 
(\ref{Prop8_decomp4_Diff_W}), 
and (\ref{Prop8_decomp6_Diff_W}), 
we get $J_1\longrightarrow 0$ as $\delta \downarrow 0$, 
hence we arrive at $\lim_{\delta \downarrow 0}\E[ \vert W_t - \tilde{W}_t \vert ^{1/2} ] = 0$. 
Therefore, by Lemma \ref {lemm_conti_u} we obtain 
\begin{align*}
&\E[u(W_t,\varphi_t,\exp(X_t))] -V_t(w,\varphi,s;u)\\
&\quad \leq \lim_{\delta \downarrow 0}|\E[u(W_t,\varphi_t,\exp(X_t))] - 
\E[u(\tilde{X}_t,\tilde{\varphi }_t,\exp(\tilde{X}_t))]| = 0. 
\end{align*}
Since $(\zeta_r)_{0\leq r\leq t} \in \mathcal{A}_t(\varphi - \psi )$ is arbitrary, 
we get 
\begin{eqnarray*}
V_t(w+F(\psi)s, \varphi - \psi, 
s e^{-\gamma h(\infty) \psi};u)\leq V_t(w,\varphi,s;u). 
\end{eqnarray*}
for an arbitrary $\psi \in [0, \varphi ]$. 
Now we complete the proof of (\ref{obj_t_ineq})\,.  
\end{proof}

\begin{proof}[Proof of Lemma \ref {temp_tilde_u}]
We may assume $\gamma > 0$. 
Fix any $\varepsilon \in (0, 1)$ and set $\varepsilon ' =\gamma \varepsilon / (1 - \varepsilon ) $. 
By (\ref {conv_L}), we see that 
for almost all $\omega $, 
there exists a $\delta _0 = \delta _0(\omega ) > 0$ such that 
$L_\delta / \delta < \gamma + \varepsilon '$ for each 
$\delta \in (0, \delta _0)$. 
Let $\delta _1 = \delta _1(\omega ) = (1 + \varepsilon ' / \gamma )^{-1}\delta _0$ and 
take any $\delta \in (0, \delta _1)$. 
Moreover, let $\delta ' = (1 + \varepsilon ' / \gamma )^{-1}\delta $. 
Then we see that $\delta ' < \delta _0$ and thus 
$L_{\delta '} < (\gamma  + \varepsilon ')\delta ' = \gamma \delta $. 
By this inequality and the definition of $\tilde{u}(\delta )$, we get 
$1 \geq \tilde{u}(\delta ) / \delta  \geq \delta ' / \delta = 1 - \varepsilon $, 
which implies the assertion. 
\end{proof}

\subsection{Proof of Theorem \ref {th_LL}}\label{sec_proof_th_LL}
We can confirm assertion (i) by applying It\^o's formula to $\overline{S}_r$ and $\overline {W}_r$. 
By a similar argument to that in Section 7.9 in \cite {Kato}, we obtain 
\begin{eqnarray*}
\E [U(\overline {W}_t)] 
&\leq & 
U\Bigg( \bar{w} + 
\int ^t_0\E \Bigg [ \frac{1 - e^{-\gamma \alpha _0\overline {\varphi }_r}}{\gamma \alpha _0}
\hat{b}(\overline {S}_re^{\gamma \alpha _0\overline {\varphi }_r})\\
&&\hspace{23mm} - 
\int _{(0, \infty )}
\frac{e^{\gamma \alpha _0\overline {\varphi }_r} - 1}{\gamma \alpha _0}
\overline {S}_r(1-e^{-\alpha _0\zeta _rz})\nu (dz) \Bigg]dr\Bigg) 
\end{eqnarray*}
for any $(\overline {\varphi }_r)_r\in \overline {\mathcal {A}}_t(\varphi )$ 
by virtue of the Jensen inequality. 
Since $\hat{b}$ is non-positive, 
the function $U$ is non-decreasing, and the terms
\begin{eqnarray*}
1 - e^{-\gamma \alpha _0\overline {\varphi }_r}, \ 
e^{\gamma \alpha _0\overline {\varphi }_r} - 1, \ 
1-e^{-\alpha _0\zeta _rz}
\end{eqnarray*}
are all non-negative, we see that 
$\E [U(\overline {W}_t)]\leq U(\overline {w})$ for any $(\overline {\varphi }_r)_r\in \overline {\mathcal {A}}_t(\varphi )$, 
which implies $\overline {V}^\varphi _t(\bar{w}, \bar{s})\leq U(\bar{w})$. 
The opposite inequality $\overline {V}^\varphi _t(\bar{w}, \bar{s})\geq U(\bar{w})$ is obtained 
similarly to the result in Section 7.9 in \cite {Kato}. 
This completes the proof. 

\subsection {Proof of Proposition \ref {th_comp_noise}} \label{sec_proof_th_comp_noise} 

The following proposition immediately leads us to (\ref {comp_noise}). 

\begin{proposition} \ 
$V^n_k(w, \varphi , s ; u_\mathrm {RN}) \geq \bar{V}^n_k(w, \varphi , s ; u_\mathrm {RN})$, 
where $V^n_k$ is defined as in \cite {Ishitani-Kato_COSA1} and $\bar{V}^n_k$ is 
obtained from $V^n_k$ by replacing $c^n_k$ with $\tilde{\gamma }$. 
\end{proposition}

\begin{proof}
We use the notation of \cite {Ishitani-Kato_COSA1}. 
Take any $(\psi ^n_l)_l\in \mathcal {A}^n_k(\varphi )$ and 
let $(W^n_l, \varphi ^n_l, S^n_l)_l = \Xi ^n_k(w, \varphi , s ; (\psi ^n_l)_l)$ be 
the triplet for $\bar{V}^n_k(w, \varphi , s ; u_{\mathrm {RN}})$. 
Since $c^n_l$ is independent of $\mathcal {F}^n_l$, 
the Jensen inequality implies 
\begin{align*}
\E [W^n_k] 
&=
w + \sum ^{k-1}_{l = 0}\E [\psi ^n_lS^n_l\exp (-\E [c^n_l | \mathcal {F}^n_l]g_n(\psi ^n_l))]\\
&\leq  
w + \sum ^{k-1}_{l = 0}\E [\psi ^n_lS^n_l\E [\exp (-c^n_lg_n(\psi ^n_l)) | \mathcal {F}^n_l]] \leq 
V^n_k(w, \varphi , s ; u_\mathrm {RN}). 
\end{align*}
Since $(\psi ^n_l)_l$ is arbitrary, we obtain the assertion. 
\end{proof}

\par\bigskip\noindent
{\bf Acknowledgment.} 
The authors are grateful to Prof.~Tai-Ho Wang (Baruch College, The City University of New York) 
for helpful comments and discussions on the subject matter. 
In addition, the authors thank the reviewer for various comments 
and constructive suggestions to improve the quality of the paper. 

\bibliographystyle{amsplain}

\end{document}